\documentclass[runningheads]{llncs}
\usepackage{graphicx}
\usepackage[nointegrals]{wasysym}

\usepackage{microtype}
\usepackage{graphicx}
\usepackage{subfigure}
\usepackage{booktabs}

\usepackage{amsmath}

\usepackage{amsthm}

\usepackage{amsfonts}
\usepackage{amssymb}
\usepackage{mathtools}
\usepackage{verbatim}
\usepackage[inline]{enumitem}
\usepackage{relsize}

\DeclarePairedDelimiter{\ceil}{\lceil}{\rceil}

\newcommand{\deltahat}{\hat{\delta}}

\newcommand{\prs}{\mathbf{P}}
\newcommand{\prstuple}{\langle \Sigma,P,P_c,R\rangle}

\newcommand{\defined}{\mathrm{def}}
\newcommand{\tftos}[1]{S_{#1}} 
\newcommand{\repwith}{\leftarrow} 
\newcommand{\rep}[2]{_{[#1\repwith #2]}}

\newcommand{\dfa}{A}
\newcommand{\complete}[1]{#1_C}
\newcommand{\pdfa}{p}
\newcommand{\pdfai}[1]{\pdfa^{#1}}
\newcommand{\pinit}{\pdfai{I}}
\newcommand{\dfai}[1]{\dfa_{#1}}
\newcommand{\rhoi}[1]{\rho^{(#1)}}
\newcommand{\gi}{\dfai{i}}
\newcommand{\gp}{\dfai{i+1}}
\newcommand{\compose}{\circ}
\newcommand{\dfatuple}{\langle\Sigma,q_0,Q,F,\delta\rangle}
\newcommand{\pdfatuple}{\langle\Sigma,q_0,Q,q_X,\delta\rangle}

\newcommand{\tree}{{\cal{T}}}
\newcommand{\treei}[1]{{\cal{T}}^{(#1)}}

\newcommand{\phat}{{\hat{p}}}
\newcommand{\phati}[1]{\hat{p}^{#1}}
\newcommand{\zphat}{{\hat{Z_p}}}

\newcommand{\initset}{\mathcal{I}}
\newcommand{\initseti}[1]{\initset_{#1}}
\newcommand{\pairsymbol}{\mathcal{D}}
\newcommand{\pairsymboli}[1]{\pairsymbol_{#1}}
\newcommand{\explicitpairi}[1]{\langle \dfai{#1}, \initseti{#1} \rangle}
\newcommand{\rulearrow}{\twoheadrightarrow} 
\newcommand{\pairname}{{enabled DFA}}
\newcommand{\pdfasdfa}{\dfa^{\pdfa}}
\newcommand{\pdfasdfai}[1]{\dfa^{\pdfai{#1}}}

\newcommand{\imap}{\mathcal{I}}
\newcommand{\imapi}[1]{\imap_{#1}}

\newcommand{\explicitpairii}[1]{\langle \dfai{#1}, \jmapi{#1} \rangle}
\newcommand{\undefined}{\bot}

\makeatletter
\newcommand{\shorteq}{
  \settowidth{\@tempdima}{-}
  \resizebox{\@tempdima}{\height}{=}%
}
\newcommand{\connect}{{\mathlarger\circ} \! \mathrm{\shorteq}\ }

\newcommand{\concat}{{\cdot}}

\newcommand{\lstar}{{$L^*$}~}

\newcommand{\patternpair}{\langle P, P_c \rangle }

\newcommand{\joinf}{\mathrm{join}}

\newcommand{\nth}{\textsuperscript{th}}

\usepackage{xcolor}

\makeatletter
\RequirePackage[bookmarks,unicode,colorlinks=true]{hyperref}%
   \def\@citecolor{blue}%
   \def\@urlcolor{blue}%
   \def\@linkcolor{blue}%

\def\orcidID#1{\smash{\href{http://orcid.org/#1}{\protect\raisebox{-1.25pt}{\protect\includegraphics{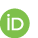}}}}}
\makeatother

\begin{document}
\title{Synthesizing Context-free Grammars from Recurrent Neural Networks} 
\subtitle{Extended Version\thanks{This is an extended version of a paper that will appear in the
27th International Conf on Tools and Algorithms for the Construction and Analysis of Systems (TACAS 2021)}}

\author{Daniel M. Yellin\inst{1}\orcidID{0000-0001-7214-5610} \and
Gail Weiss\inst{2}\orcidID{0000-0003-0762-9090}}

\authorrunning{D.M. Yellin G. Weiss}
\institute{IBM, Givatayim, Israel \\ \email{dannyyellin@gmail.com} \and
Technion, Haifa, Israel \\ 
\email{sgailw@cs.technion.ac.il}\\}
\maketitle
\setcounter{footnote}{0}
\begin{abstract}
We present an algorithm for extracting a subclass of the context free grammars (CFGs) from a trained recurrent neural network (RNN).
We develop a new framework, \emph{pattern rule sets} (PRSs), which describe sequences of deterministic finite automata (DFAs) that approximate a non-regular language. 
We present an algorithm for recovering the PRS behind a sequence of such automata, and apply it to the sequences of automata extracted from trained RNNs using the \lstar algorithm.
We then show how the PRS may converted into a CFG, enabling a familiar and useful presentation of the learned language. 

Extracting the learned language of an RNN is important to facilitate understanding of the RNN and to verify its correctness.   
Furthermore, the extracted CFG can augment the RNN in classifying correct sentences, as the RNN's predictive accuracy decreases when the recursion depth and distance between matching delimiters of its input sequences increases.

\end{abstract}

\section{Introduction}

Recurrent Neural Networks (RNNs) are a class of neural networks adapted to sequential input, enjoying wide use in a variety of sequence processing tasks.
Their internal process is opaque, prompting several works into extracting interpretable rules from them.
Existing works focus on the extraction of deterministic or weighted finite automata (DFAs and WFAs) from trained RNNs \cite{omlin-giles-96,NNExtractionFuzzyClustering,WeissGoldbergYahav2018,WFAfromRNNSpectral}.

However, DFAs are insufficient to fully capture the behavior of RNNs, which are known to be theoretically Turing-complete \cite{siegelmann95}, and for which there exist architecture variants such as LSTMs \cite{LSTM} and features such as stacks \cite{Das-stack-92,Sun1997} or attention \cite{bahdanau-attention} increasing their practical power.
Several recent investigations 
explore the ability of different RNN architectures to learn Dyck, counter, and other non-regular languages \cite{sennhauser-berwick-2018-evaluating,Bernardy2018,YuVukuhn2019,skachkova-etal-2018-closing}, with mixed results. While the data indicates that RNNs can generalize and achieve high accuracy, they do not learn 
hierarchical rules, and generalization deteriorates as the distance or depth between matching delimiters becomes dramatically larger\cite{sennhauser-berwick-2018-evaluating,Bernardy2018,YuVukuhn2019}.   
Sennhauser and Berwick conjecture that 
``what the LSTM has in fact acquired is sequential statistical approximation to this 
solution" instead of ``the `perfect' rule-based solution" \cite{sennhauser-berwick-2018-evaluating}.  Similarly, Yu et. al. conclude that ``the RNNs 
can not truly model CFGs, even when powered by the attention mechanism".

\paragraph{Goal of this paper}
We wish to extract a CFG from a trained RNN.
Our motivation is two-fold: first, extracting a CFG from the RNN 
is important to facilitate understanding of the RNN and to verify its correctness.
Second, the learned CFG may be used to augment or generalise the rules learned by the RNN, whose own predictive ability decreases as the depth of nested structures and distance between matching constructs in the input sequences increases \cite{Bernardy2018,sennhauser-berwick-2018-evaluating,YuVukuhn2019}.
Our technique can synthesize the CFG based upon training data with relatively short distance and small depth.  As pointed out in \cite{hewitt-etal-2020-rnns}, 
a fixed precision RNN can only learn a language of fixed depth strings (in contrast to an idealized infinite precision RNN that can recognize any Dyck language\cite{Korsky2019}).  
Our goal is to find the CFG that not only
explains the finite language learnt by the RNN, but generalizes it to strings of unbounded depth and distance.  

\paragraph{Our approach}
Our method builds on the DFA extraction work of Weiss et al. \cite{WeissGoldbergYahav2018}, which uses the \lstar algorithm \cite{Lstar} to learn the DFA of a given RNN.
The \lstar algorithm operates by generating a
\emph{sequence} of DFAs, each one a hypothesis for the target language, and interacting with a teacher, in our case the RNN, to improve them.
Our main insight is that we can view these DFAs as increasingly accurate approximations of the target CFL.  
We assume that each hypothesis improves on its predecessor by applying an unknown rule that recursively increases the distance and embedded depth of sentences accepted by the underlying CFL.  
In this light, synthesizing the CFG responsible for the language learnt by the RNN
becomes the problem of recovering these rules.
A significant issue
we must also address is that the DFAs produced are often
inexact or not as we expect, either due to the failure of the RNN to accurately learn the language, or as an artifact of the \lstar algorithm.  

\begin{figure*}[t]
\includegraphics[scale =.4,trim= 30mm 80mm 20mm 50mm,clip=true]{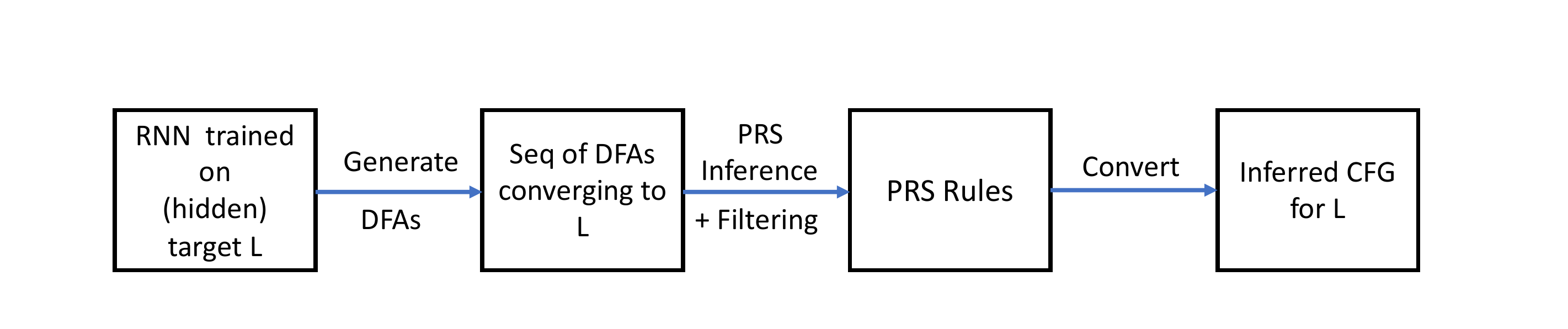}
\caption{Overview of steps in algorithm to synthesize the hidden language $L$}
\label{fig:overview}
\end{figure*}

We propose the framework of \emph{pattern rule sets} (PRSs) for describing such rule applications, and present an algorithm for recovering a PRS from a sequence of DFAs. We also provide a method for converting a PRS to a CFG, translating our extracted rules into familiar territory. We test our method on RNNs trained on several PRS languages.

Pattern rule sets are expressive enough to cover several variants of the Dyck languages, which are prototypical CFLs: 
the Chomsky–Schützenberger representation theorem shows that any context-free language can be expressed as a homomorphic image
of a Dyck language intersected with a regular language\cite{Kozen1977}.

To the best of our knowledge, this is the first work on synthesizing a CFG from a general RNN  \footnote{Though some works extract push-down automata \cite{Sun1997,Das-stack-92} from RNNs  with an external stack (Sec. \ref{Sect:related}), they do not apply to plain RNNs.
}.

\paragraph{Contributions}
The main contributions of this paper are:
\begin{itemize}
    \item \emph{Pattern Rule Sets} (PRSs), a framework for describing a sequence of DFAs approximating a CFL.
    \item An algorithm for recovering the PRS generating a sequence of DFAs, that may also be applied to noisy DFAs elicited from an RNN using \lstar.
    \item An algorithm converting a PRS to a CFG.
    \item An implementation of our technique, and an evaluation of its success on recovering various CFLs from trained RNNs. \footnote{The implementation for this paper, and a link to all trained RNNs, is available at {\tt https://github.com/tech-srl/RNN\_to\_PRS\_CFG}.}

\end{itemize}
The overall steps in our technique are given in Figure \ref{fig:overview}.  
The rest of this paper is as follows.  Section \ref{Sec:Background} provides basic definitions used in the paper, and
Section \ref{Sect:Patterns} introduces {\it Patterns}, a restricted form of DFAs.  Section \ref{sect:PRS} defines
{\em Pattern Rule Sets (PRS)}, the main construct of our research. Section \ref{Sect:rule-inference} gives an 
algorithm to recover a PRS from a sequence of DFAs, even in the presence of noise, and Section
\ref{Sect:expressibility} gives an algorithm to convert a PRS into a CFG.   Section \ref{sect:experiments} presents our experimental
results, Section \ref{Sect:related} discusses related research and Section \ref{Sect:future} outlines directions for future research. 
Appendices \ref{sect:appendixObs},\ref{Sect:appendixCorrect} and \ref{Sect:AppendixCFG} provide proofs of the correctness of the algorithms given 
in the paper, as well results relating to the expressibility of a PRS.

\section{Definitions and Notations}\label{Sec:Background}

\subsection{Deterministic Finite Automata}
\begin{definition}[Deterministic Finite Automata]
A deterministic finite automaton (DFA) over an alphabet $\Sigma$ is a 5-tuple $\langle\Sigma,q_0,Q,F,\delta \rangle$ such that $Q$ is a finite set of states, $q_0\in Q$ is the initial state, $F\subseteq Q$ is a set of final (accepting) states and $\delta:Q\times\Sigma\rightarrow Q$ is a (possibly partial) transition function. 
\end{definition}

Unless stated otherwise, we assume each DFA's states are unique to itself, i.e., for any two DFAs $A,B$ -- including two instances of the same DFA --  $Q_A\cap Q_B=\emptyset$. A DFA $\dfa$ is said to be {\em complete} if $\delta$ is complete, i.e., the value $\delta(q,\sigma)$ is defined for every $q,\sigma\in Q\times \Sigma$. Otherwise, it is \emph{incomplete}.

We define the extended transition function
$\deltahat:Q\times\Sigma^*\rightarrow Q$  and the language $L(\dfa)$ accepted by $\dfa$ in the typical fashion. We also associate a language with intermediate states of $\dfa$: $L(\dfa,q_1,q_2)\triangleq \{w\in\Sigma^*\ |\ \deltahat(q_1,w)=q_2\}$. 
The states from which no sequence $w\in\Sigma^*$ is accepted are known as the \emph{sink reject states}.

\begin{definition}
The \emph{sink reject states} of a DFA $\dfa=\dfatuple$ are the maximal set $Q_R\subseteq Q$ satisfying: $Q_R\cap F=\emptyset$, and for every $q\in Q_R$ and $\sigma\in\Sigma$, either $\delta(q,\sigma)\in Q_R$ or $\delta(q,\sigma)$ is not defined.
\end{definition}

\emph{Incomplete} DFAs are partial representations of complete DFAs, where every unspecified transition 
is shorthand for a transition to a sink reject state.
All definitions for complete DFAs are extended to incomplete DFAs $\dfa$ by considering their \emph{completion}: the DFA $\complete{\dfa}$ obtained by connecting a (possibly new) sink reject state to all its missing transitions.

\begin{definition} [Defined Tokens]
Let $\dfa=\dfatuple$ be a complete DFA with sink reject states $Q_R$. 
For every $q\in Q$, its
\emph{defined tokens} are
$\defined(\dfa,q)\triangleq\{\sigma\in\Sigma\ |\ \delta(q,\sigma) \notin Q_R\}$. When the DFA $\dfa$ is clear from context, we write $\defined(q)$.
\end{definition}

We now introduce terminology that will help us discuss merging automata states.

\begin{definition} [Set Representation of $\delta$]
A (possibly partial) transition function $\delta:Q\times\Sigma\rightarrow Q$ may be equivalently defined as the set $\tftos{\delta}=\{(q,\sigma,q')\ |\ \delta(q,\sigma)=q'\}$. We use $\delta$ and $S_\delta$ interchangeably.
\end{definition}

\begin{definition} [Replacing a State]
For a transition function $\delta:Q\times\Sigma\rightarrow Q$, state $q\in Q$, and new state $q_n\notin Q$, we denote by
$\delta\rep{q}{q_n}:
Q'\times\Sigma\rightarrow Q'$ the transition function over $Q'=(Q\setminus\{q\})\cup\{q_n\}$
and  $\Sigma$ that is identical to $\delta$ except that it redirects all transitions 
into or out of $q$ to be into or out of $q_n$. 
\end{definition}

\subsection{Dyck Languages}
\label{Sect:Dyck}
A Dyck language \emph{of order $N$} is expressed by the grammar
\begin{tt}
D ::= $\varepsilon$  | L$_i$ D R$_i$ | D D 
\end{tt} with start symbol {\tt D},
where for each $1\leq i\leq N$,  L$_i$ and R$_i$ are matching left and right delimiters. 
A common methodology for measuring the complexity of a Dyck word is to measure its maximum \emph{distance} (number of characters) and \emph{embedded depth} (maximum number of unclosed delimiters) between matching delimiters~\cite{sennhauser-berwick-2018-evaluating}.

While  L$_i$ and R$_i$ are single characters in a Dyck language, we generalize and refer to {\em Regular Expression Dyck  (RE-Dyck)} languages  as 
languages expressed by the same CFG,
except that each {\tt L}$_i$ and each {\tt R}$_i$ 
derive some regular expression.

 \paragraph{Regular Expressions:} We present regular expressions as is standard, 
 for example:
 $\{ a | b \} \concat c$ refers to the language consisting of one of $a$ or $b$, followed by $c$.
\section{Patterns}
\label{Sect:Patterns}
Patterns are DFAs with a single {\em exit} state $q_X$ in place of a set of final states,
and with no cycles on their initial or exit states unless $q_0=q_X$.  In this paper we express patterns in incomplete representation, i.e., they have no explicit sink-reject states.

\begin{definition}[Patterns]
A \emph{pattern} $\pdfa=\pdfatuple$ is a DFA $\pdfasdfa=\langle \Sigma, q_0, Q, \{q_X\}, \delta \rangle$ 
satisfying: $L(A^p)\neq \emptyset$, and either $q_0=q_X$, or $\defined(q_X)=\emptyset$ and $L(\dfa,q_0,q_0)=\{\varepsilon\}$.  
If $q_0 = q_X$ then $p$ is called \emph{circular}, otherwise, it is \emph{non-circular}.
\label{Defn:pattern}
\end{definition}
Note that our definition does not rule out a cycle in the middle of an {\em non-circular} pattern but only one that traverses the initial or final states.

All the definitions for DFAs 
apply to patterns through $\pdfasdfa$. 
We denote each pattern $p$'s language
$L_p\triangleq L(p)$,
and if it is marked by some superscript $i$, we refer to all of its components with superscript $i$: $p^i=\langle \Sigma, q_0^i,Q^i,q_X^i,\delta^i\rangle$.

\subsection{Pattern Composition}
We can compose two non-circular patterns $\pdfa^1,\pdfa^2$ 
by merging the exit state of $\pdfa^1$ with the initial state of $\pdfa^2$,
creating a new
pattern $\pdfa^3$ satisfying 
$L_{\pdfa^3}=L_{\pdfa^1}\concat L_{\pdfa^2}$.

\begin{definition}[Serial Composition]
Let $p^1,p^2$
be two 
non-circular
patterns.
Their {\em serial composite} is the pattern $p^1\compose p^2 = \langle \Sigma, q_0^1,Q,q_X^2,\delta\rangle$ in which
$Q=Q^1\cup Q^2\setminus\{q_X^1\}$ and
$\delta = \delta^1\rep{q_X^1}{q_0^2} \cup \delta^2$.  
We call $q_0^2$ the {\em join state} of this operation.
\end{definition}

If we additionally merge the exit state of $\pdfa_2$ with the initial state of $\pdfa_1$, 
we obtain a circular pattern $\pdfa$ which we call the \emph{circular composition} of $p_1$ and $p_2$. This composition satisfies $L_{\pdfa} = \{L_{\pdfa_1}\concat L_{\pdfa_2}\}^*$.

\begin{definition}[Circular Composition]
Let $p^1,p^2$
be two non-circular patterns.
Their {\em circular composite} is the circular pattern $p_1\compose_c p_2 = \langle \Sigma, q_0^1,Q,q_0^1,\delta\rangle$ in which
$Q=Q^1\cup Q^2\setminus\{q_X^1,q_X^2\}$ and $\delta = \delta^1\rep{q_X^1}{q_0^2} \cup \delta^2\rep{q_X^2}{q_0^1}$.   
We call $q_0^2$ the {\em join state} of this operation.
\end{definition}

Figure \ref{fig:composition-examples} shows 3 examples of serial and circular compositions of patterns.
\begin{figure*}
\includegraphics[scale = .33,trim = -30mm 11mm 0mm 12mm]{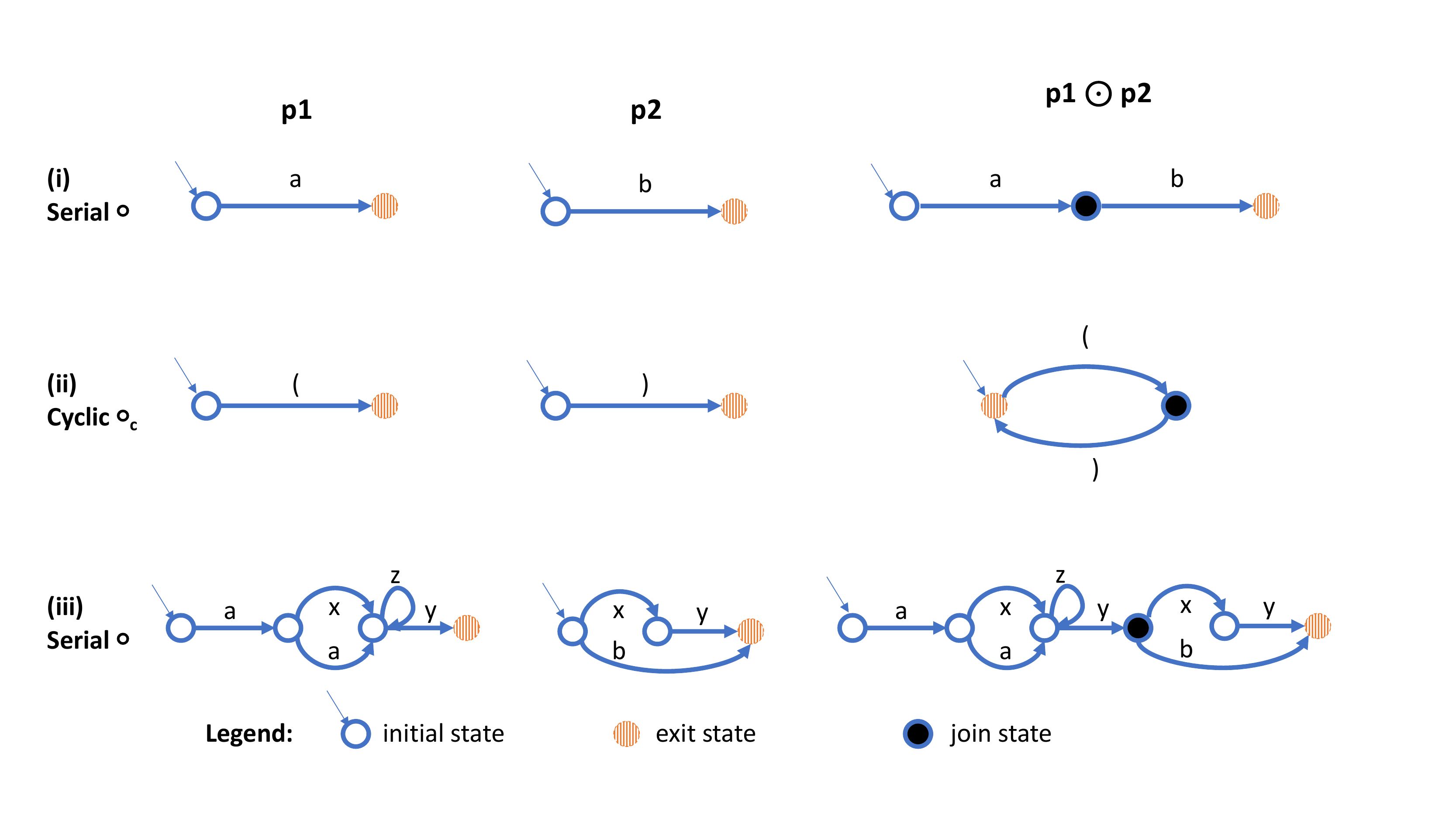}
\caption{Examples of the composition operator}
\label{fig:composition-examples}
\end{figure*}

Patterns do not 
carry information about whether or not they have been composed from other patterns. 
We maintain such information using \emph{pattern pairs}.

\begin{definition}[Pattern Pair]
A \emph{pattern pair} is a pair $\langle P,P_c\rangle$ of pattern sets, such that $P_c\subset P$  and for every $p\in P_c$ there exists exactly one pair $p_1,p_2\in P$ satisfying $p = p_1\odot p_2$ for some $\odot\in\{\compose,\compose_c\}$.
We refer to the patterns $p\in P_c$ as the \emph{composite patterns} of $\langle P,P_c\rangle$, and to the rest as its \emph{base patterns}.
\end{definition}

Every instance $\phat$ of a pattern $p$ in a DFA $\dfa$ is uniquely defined by $p$, $\dfa$, and $\phat$'s initial state in $\dfa$. If $p$ is a composite pattern with respect to some pattern pair $\patternpair$, 
the join state of its composition within $\dfa$ is also uniquely defined.

\begin{definition}[Pattern Instances]
Let $\dfa=~\langle \Sigma, q^\dfa_0, Q^\dfa, F, \delta^\dfa\rangle$ be a DFA, $p=\pdfatuple$ be a pattern, and  $\phat=\langle \Sigma, q'_0, Q', q'_X, \delta'\rangle$ be a pattern `inside' $\dfa$, i.e., $Q'\subseteq Q^\dfa$ and $\delta'\subseteq \delta^\dfa$. We say that $\phat$ is an \emph{instance of $p$ in $\dfa$} if 
$\phat$ is isomorphic to $p$.
\end{definition}

A pattern instance $\phat$ in a DFA $\dfa$ is uniquely determined by its structure and initial state: $(p,q)$.

\begin{definition} For every pattern pair $\patternpair$ we define the function
$\joinf$ as follows: for each composite pattern 
$p\in P_c$,
DFA $\dfa$, and initial state $q$ of an instance $\phat$ of $p$ in $\dfa$,
$\joinf(p,q,\dfa)$ returns the join state of $\phat$ with respect to its composition in $\patternpair$. 
\end{definition}

\renewcommand{\explicitpairii}[1]{\langle \dfai{#1}, \imapi{#1} \rangle}

\section{Pattern Rule Sets}
\label{sect:PRS}
For any infinite sequence $S=\dfa_1,\dfa_2,...$ of DFAs satisfying $L(\dfa_i)\subset L(\dfa_{i+1})$, for all $i$,
we define the language of $S$ as the union of the languages of all these DFAs: $L(S)=\cup_i L(\dfa_i)$. Such sequences may be used to express CFLs
such as the language $L=\{\mathrm{a}^n\mathrm{b}^n\ |\ n\in\mathbb{N}\}$ and the Dyck language of order N.

In this work we take a finite sequence $\dfa_1,\dfa_2,...,\dfa_n$ of DFAs, and assume it is a (possibly noisy) finite prefix of an infinite sequence of approximations for a language, as above.
We attempt to reconstruct the language by guessing how the sequence may continue. 
To allow such generalization, we must make assumptions about how the sequence is generated.
For this we introduce \emph{pattern rule sets}.

Pattern rule sets (PRSs) create sequences of DFAs with a single accepting state. 
Each PRS is built around a pattern pair $\patternpair$, and each rule application involves the connection of a new pattern instance to the current DFA $\gi$, at the join state of a composite-pattern inserted whole at some earlier point in the DFA's creation. 
In order to define where a pattern can be inserted into a DFA, we introduce an \emph{enabled instance} set $\imap$.

\begin{definition}
An \emph{\pairname} over a pattern pair $\patternpair$
is a tuple $\langle \dfa,\imap \rangle$ 
such that 
$\dfa=\langle \Sigma, q_0, Q, F, \delta\rangle$
is a DFA and $\imap \subseteq P_c\times Q$ marks
\emph{enabled 
instances} of 
composite 
patterns in $\dfa$.
\end{definition}

Intuitively, for every \pairname $\langle \dfa,\imap \rangle$ and $(p,q)\in \imap$, we know: (i) there is an instance of pattern $p$ in $\dfa$
starting at state $q$, and (ii) this instance is \emph{enabled}; i.e., we may connect new pattern instances to its join state 
$\joinf (p,q,\dfa)$.

We now formally define pattern rule sets and how they are applied to create 
{\pairname}s, 
and so sequences of DFAs.

\begin{definition}
A PRS $\prs$ is a tuple
$\langle \Sigma, P, P_c, R\rangle $  where $\patternpair$ is a pattern pair over the alphabet $\Sigma$ and $R$ is a set of \emph{rules}. Each rule has one of the following forms, for some $p,p^1, p^2, p^3, \pinit \in P$, with $p^1$ and $p^2$ non-circular:
\begin{enumerate} [label=\textnormal{(\arabic*)}]
\item $\perp \rulearrow \pinit$ \label{rule-start} 
\item $p \rulearrow_c (p^1 \odot p^2) \connect p^3$, 
where $p= p^1 \odot p^2$ for $\odot \in \{\compose,\compose_c\}$, and $p^3$ is circular \label{rule-composite-circular}
\item $p \rulearrow_s (p^1 \compose p^2) \connect p^3$, where $p= p^1 \compose p^2$ and $p^3$ is non-circular \label{rule-composite} \label{rule-composite-serial} 

\end{enumerate}
\end{definition} 

A PRS is used to derive sequences of \pairname s as follows:
first, a rule of type \ref{rule-start} is used to create an initial \pairname~$\pairsymboli{1}=\explicitpairii{1}$.
Then, for any $\explicitpairii{i}$, each of the rule types define options to graft new pattern instances onto states in $\dfai{i}$, with $\imapi{i}$ determining which states are eligible to be expanded in this way.
The first DFA is simply the $\pinit$ from a rule of type \ref{rule-start}.  If $\pinit$ is composite, then it is also enabled.

\begin{definition}[Initial Composition]
$\pairsymboli{1}=\explicitpairii{1}$ is generated from a rule $\perp \rulearrow \pinit$ 
as follows: $\dfai{1}=\pdfasdfai{I}$, and $\imapi{i}=\{(\pinit,q_0^I)\}$ if $\pinit\in P_c$ and otherwise $\imapi{1}=\emptyset$. 
\label{defn:start-rule}
\end{definition}

Let $\pairsymboli{i}=\explicitpairii{i}$ be an \pairname~generated from some given PRS $\prs=\prstuple$, and denote $\dfai{i}=\dfatuple$. 
Note that for $A_1$, $|F|=1$, and we will see that $F$ is unchanged by all further rule applications. Hence we denote $F=\{q_f\}$ for all $A_i$.

Rules of type \ref{rule-start} extend $\dfai{i}$ by grafting a circular pattern to $q_0$, and then enabling that pattern if it is composite. 

\begin{definition}[Rules of type \ref{rule-start}]
	A rule $\perp \rulearrow \pinit$ with circular $\pinit$
	may extend $\explicitpairii{i}$ at the initial state $q_0$ of $\dfai{i}$ iff 
	$\defined(q_0)\cap\defined(q_0^I)=\emptyset$. 
	This creates the DFA
	$\dfai{i+1}=\langle \Sigma, q_0, Q\cup Q^I\setminus\{q_0^I\} , F, \delta\cup\delta^I\rep{q_0^I}{q_0}\rangle$. 
	If $\pinit\in P_c$ then $\imapi{i+1}=\imapi{i}\cup\{(\pinit,q_0)\}$, else $\imapi{i+1}=\imapi{i}$.
	\label{defn:start-rule-existing-dfa}
\end{definition}

Rules of type \ref{rule-composite-circular} graft a circular pattern $p^3=\langle\Sigma,q_0^3,q_x^3,F,\delta^3\rangle$ onto the join state $q_j$ of an enabled pattern instance $\phat$ in $\dfai{i}$, by merging $q_0^3$ with $q_j$. In doing so, they also enable the patterns composing $\phat$, provided they themselves are composite patterns.

\begin{definition} [Rules of type \ref{rule-composite-circular}]
A rule $p\rulearrow_c (p^1 \odot p^2) \connect p^3$ 
may extend $\explicitpairii{i}$ at 
the join state $q_j=\joinf(p,q,\dfai{i})$ of any instance $(p,q)\in\imapi{i}$,
provided
$\defined(q_j)\cap\defined(q_0^3)=\emptyset$. 
This creates 
$\explicitpairi{i+1}$
as follows:
$\dfai{i+1}=\langle \Sigma, q_0, Q\cup Q^3\setminus{q_0^3} , F, \delta\cup\delta^3\rep{q_0^3}{q_j}\rangle$, and 
$\imapi{i+1}=\imapi{i}\cup\{(p^k,q^k)\ |\ p^k\in P_c, k\in\{1,2,3\}\}$, where $q^1=q$ and $q^2=q^3=q_j$.
\label{defn:cyclic-rule}
\end{definition}

For an application of $r=p\rulearrow_c (p^1 \odot p^2) \connect p^3$, consider the languages $L_L$ and $L_R$ leading into and `back from' the considered instance $(p,q)$: $L_L=L(\dfai{i},q_0,q)$ and $L_R=L(\dfai{i},q_X^{(p,q)},q_f)$, where $q_X^{(p,q)}$
is the exit state of $(p,q)$. Where $L_L\cdot L_{p}\cdot L_R\subseteq L(\dfai{i})$, then now also 
$L_L\cdot L_{p^1}\cdot L_{p^3}\cdot L_{p_2}\cdot L_R\subseteq L(\dfai{i+1})$ (and moreover, $L_L\cdot (L_{p^1}\cdot L_{p^3} \cdot L_{p_2})^*\cdot L_R\subseteq L(\dfai{i+1})$ if $p$ is circular).
Example applications of rule \ref{rule-composite-circular} are shown in Figures \ref{fig:application-of-comp-rule}(i) and \ref{fig:application-of-comp-rule}(ii).

\begin{figure*}[t]
\includegraphics[scale = .45,trim = 20mm 00mm 00mm 0mm]{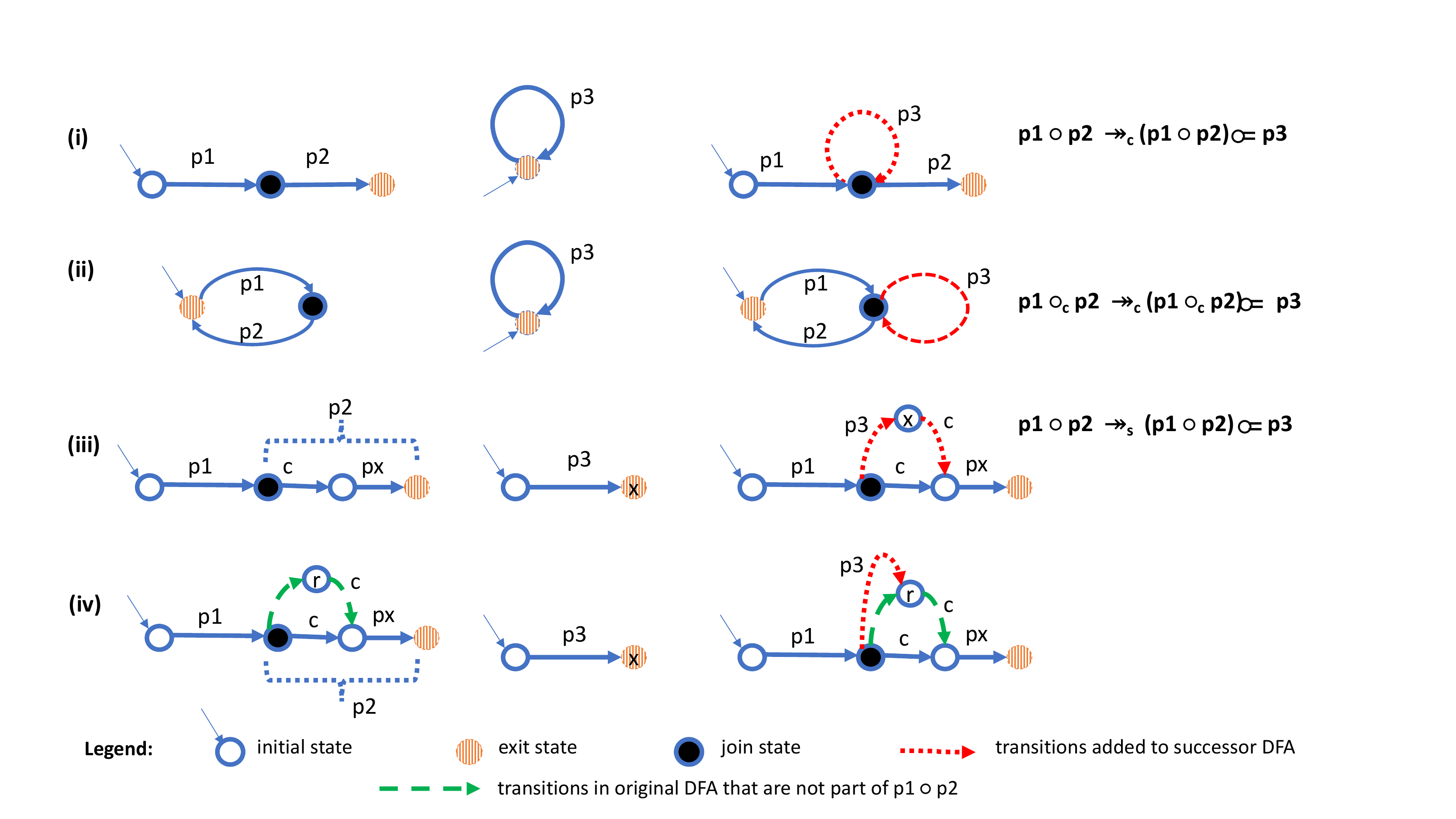}
\caption{Structure of DFA after applying rule of type 2 or type 3}
\label{fig:application-of-comp-rule}
\end{figure*}

For non-circular patterns we also wish to 
insert an optional $L_{p^3}$ between $L_{p^1}$ and $L_{p^2}$,
but this time we must avoid connecting the exit state $q_X^3$ to $q_j$ lest we loop over $p^3$ multiple times.
We therefore duplicate the outgoing transitions of $q_j$ in $p^1\compose p^2$ to the inserted state $q_X^3$ so that they may act as the connections back into the main DFA. 

\begin{definition} [Rules of type \ref{rule-composite}]
A rule $p\rulearrow_s (p^1 \compose p^2) \connect p^3$ 
may extend $\explicitpairii{i}$ at the join state 
 $q_j=\joinf(p,q,\dfa{i})$ of any instance $(p,q)\in\imapi{i}$,
provided
$\defined(q_j)\cap\defined(q_0^3)=\emptyset$.
This creates 
$\explicitpairi{i+1}$
as follows: 
$\dfai{i+1}=\langle \Sigma, q_0, Q\cup Q^3\setminus{q_0^3} , F, \delta\cup\delta^3\rep{q_0^3}{q_j}\cup C \rangle$ where 
$C=\{\ (q_X^3,\sigma,\delta(q_j,\sigma)) |\ 
\sigma\in\defined(p^{2},q_0^2) 
\}$,
and 
$\imapi{i+1}=\imapi{i}\cup\{(p^k,q^k)\ |\ p^k\in P_c, k\in\{1,2,3\}\}$ where $q^1=q$ and $q^2=q^3=q_j$.
\label{defn:acyclic-rule}
\end{definition}

We call the set $C$ \emph{connecting transitions}.
This application of this rule is depicted in Diagram (iii) of Figure \ref{fig:application-of-comp-rule}, where the transition labeled
`c' in this Diagram is a member of $C$ from our definition.  

Multiple applications of rules of type \ref{rule-composite} to the same instance $\phat$ will create several equivalent states in the resulting DFAs, as all of their exit states will have the same connecting transitions.  
These states are merged in a minimized representation,
as depicted in Diagram (iv) of Figure \ref{fig:application-of-comp-rule}.

We now formally define the language defined by a PRS. This is the language that we will assume a given finite sequence of DFAs is trying to approximate.

\begin{definition}[DFAs Generated by a PRS]
	We say that a PRS $\prs$ \emph{generates} a DFA $\dfa$, denoted $\dfa\in G(\prs)$, if there exists a finite sequence of \pairname s $\explicitpairi{1},...,\explicitpairi{i}$ obtained only by applying rules from $\prs$, for which $\dfa=\dfa_i$.
\end{definition}

\begin{definition}[Language of a PRS]
The language of a PRS $\prs$ is the union of the languages of the DFAs it can generate: $L(\prs)=\cup_{\dfa\in G(\prs)}L(\dfa)$.
\end{definition}

\subsection{Examples}
\label{sect:prs-examples}
\emph{EXAMPLE 1:}
Let $p^1$ and $p^2$ be the patterns  accepting `a' and `b' respectively.  
Consider the rule set $R_{ab}$ with two rules, 
$\perp \rulearrow p^1 \compose p^2$ and $p^1\compose p^2 \rulearrow_s (p^1 \compose p^2) \connect (p^1\compose p^2)$.
This rule set creates only one sequence of DFAs. 
Once the first rule creates the initial DFA, by
continuously applying the second rule, we obtain the infinite sequence of DFAs each satisfying $L(\dfai{i})=\{a^jb^j : 1 \le j \le i\}$, and so   
$L(R_{ab}) = \{a^ib^i : i >0 \}$.  Figure \ref{fig:composition-examples}(i)
presents $\dfai{1}$, while $\dfai{2}$ and $\dfai{3}$ appear in Figure \ref{fig:anbn}(i). Note that we can substitute any non-circular patterns for $p^1$ and $p^2$,
creating the language $\{x^iy^i : i >0 \}$ for any pair of non-circular pattern regular expressions $x$ and $y$.
\begin{figure}
\includegraphics[scale =.4,trim=00mm 70mm 00mm 5mm,clip=true]{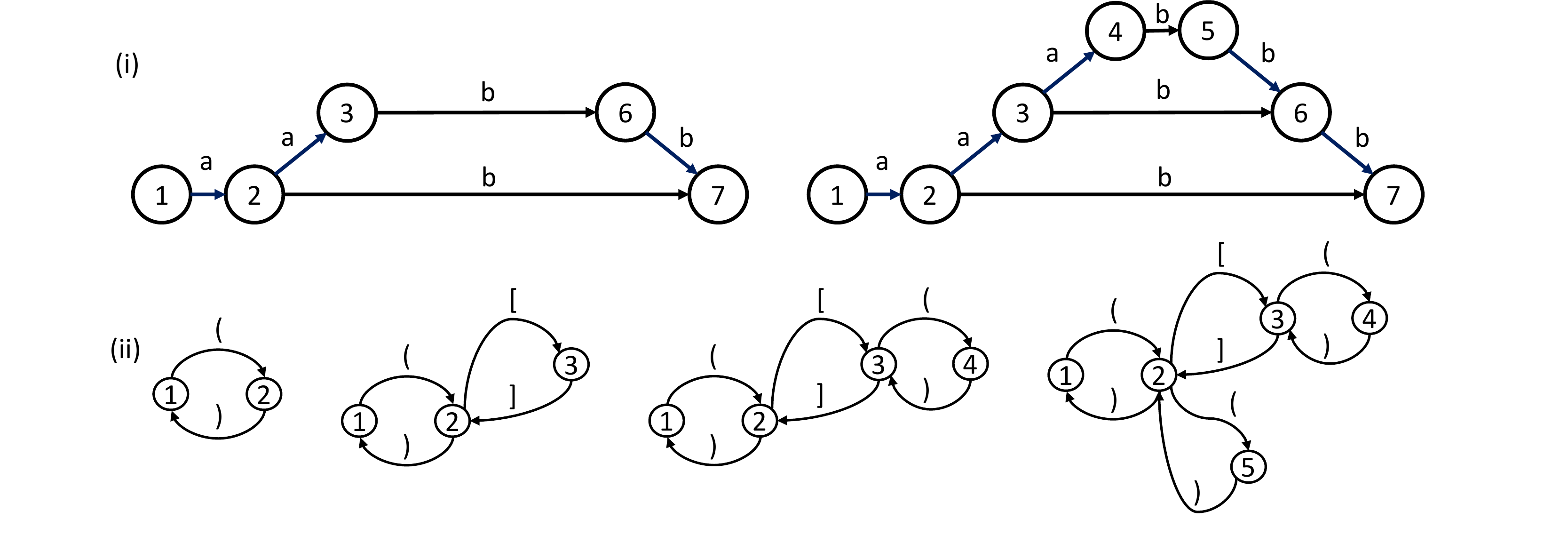}
\caption{DFAs sequences for $R_{ab}$ and $R_{Dyck2}$}
\label{fig:anbn}
\end{figure}

\emph{EXAMPLE 2:}
Let $p^1$,$p^2$,$p^4$, and $p^5$ be the non-circular patterns accepting `(', `)', `[', and `]' respectively.  Let $p^3 = p^1 \compose_c p^2$ and $p^6 = p^4 \compose_c p^5$.
Let $R_{Dyck2}$ be the PRS containing rules $\bot \rulearrow p^3$, $\bot \rulearrow p^6$, 
$p^3 \rulearrow_c \ (p^1 \compose_c p^2) \connect p^3$, 
$p^3 \rulearrow_c \ (p^1 \compose_c p^2) \connect p^6$,
$p^6 \rulearrow_c \ (p^4 \compose_c p^5) \connect p^3$, and 
$p^6 \rulearrow_c \ (p^4 \compose_c p^5) \connect p^6$.
$R_{Dyck2}$ defines the Dyck language of order $2$.  Figure \ref{fig:anbn} (ii) shows one of its possible DFA-sequences. 

\section{PRS Inference Algorithm}
\label{Sect:rule-inference}
We have shown how a PRS can generate a sequence of DFAs that can define, in the limit, a non-regular language.  However, we are interested in the dual problem:
given a sequence of DFAs generated by a PRS $\prs$, can we reconstruct $\prs$? 
Coupled with an \lstar extraction of DFAs from a trained RNN, solving this problem will enable us to extract a PRS language from an RNN, provided the \lstar extraction also follows a PRS pattern (as we often find it does).

We present an algorithm for this problem, and show its correctness in Section \ref{Sect:inference-correctness}. 
We note that in practice the DFAs we are given are not ``perfect''; they
contain noise that deviates from the PRS. We therefore augment this algorithm in Section \ref{Sect:deviations},
allowing it to operate smoothly even on imperfect DFA sequences created from RNN extraction.

In the following, for each pattern instance $\phat$ in $\dfai{i}$, we denote by $p$ the pattern that it is an instance of. Additionally, for each consecutive DFA pair $\dfai{i}$ and $\dfai{i+1}$, we refer by $\phat^3$ to the new pattern instance in $\dfai{i+1}$.
\paragraph{Main steps of inference algorithm.}
Given a sequence of DFAs $\dfai{1}\cdots \dfai{n}$, the algorithm infers $\prs= \langle \Sigma,P, P_c, R\rangle$
in the following stages:
\begin{enumerate}
\item Discover the initial pattern instance $\phat^I$ in $\dfai{1}$. Insert $p^I$ into $P$ and mark $\phat^I$ as enabled.  
Insert the rule $\bot \rightarrow p^I$ into $R$.  \label{InferAlg-step1}
\item For $i, 1 \le i \le n-1$:
\begin{enumerate}
\item Discover the new pattern instance $\phat^3$ in $\dfai{i+1}$ that extends $\dfai{i}$. \label{InferAlg-step2a}
\item If $\phat^3$ starts at the initial state $q_0$ of $\dfai{i+1}$, then it is an application of a rule of type \ref{rule-start}.
Insert $p^3$ into $P$ and mark $\phat^3$ as enabled, and add the rule $\bot \rulearrow p^3$ to $R$. 
\label{InferAlg-step2b}
\item Otherwise ($\phat^3$ does not start at $q_0$), find the unique enabled pattern $\phat = \phat^1\odot \phat^2$ in $\dfai{i}$ s.t. $\phat^3$'s initial state $q$ is the join state of $\phat$. 
Add $p^1,p^2$, and $p^3$ to $P$ and $p$ to $P_c$, and mark $\phat^1$,$\phat^2$, and $\phat^3$ as enabled.
If $\phat^3$ is non-circular add the rule  $p \rulearrow_s (p^1 \compose p^2) \connect p^3$ to $R$,  
otherwise add the rule $p \rulearrow_c (p^1 \odot p^2) \connect p^3$ to $R$. 
\label{InferAlg-step2c}
\end{enumerate}
\item Define $\Sigma$ to be the set of symbols used by the patterns $P$.   
\end{enumerate}
Once we know the newly created pattern $\pinit$ or  $\phat^3$ (step \ref{InferAlg-step1} or \ref{InferAlg-step2a}) and the pattern $\phat$ that it is grafted onto 
(step \ref{InferAlg-step2c}), creating the rule is straightforward.   We elaborate below on the how the algorithm accurately finds these patterns.   

\label{Sect:PatDiscovery}
\paragraph{Discovering new patterns $\phat^I$ and $\phat^3$}
The first pattern $\pinit$ is easily discovered; it is $\dfai{1}$, the first DFA.
To find those patterns added in subsequent DFAs, we need to isolate the pattern added between $\gi$ and $\gp$, by identifying which states in $\gp= \langle \Sigma, q_0',Q',F',\delta'\rangle$ are `new' relative to $\gi=\dfatuple$.
From the PRS definitions, we know that there is a subset of states and transitions in $\gp$ that is isomorphic to $\gi$:
\begin{definition} (Existing states and transitions) For every $q'\in Q'$, we say that $q'$ \emph{exists} in $\gi$, with \emph{parallel state} $q\in Q$, 
iff there exists a sequence $w \in \Sigma^*$ such that $q = \deltahat(q_0,w)$, $q' = \deltahat'(q_0,w)$, and neither is a sink reject state. Additionally, for every $q_1',q_2'\in Q'$ with parallel states $q_1,q_2\in Q$, we say that $(q_1',\sigma,q_2')\in\delta'$ \emph{exists} in $\gi$ if $(q_1,\sigma,q_2)\in \delta$.  
\label{def:equiv}
\end{definition}

We refer to the states and transitions in $\gp$ that do not exist in $\gi$ as the \emph{new} states and transitions of $\gp$, denoting them $Q_N\subseteq Q'$ and $\delta_N\subseteq \delta'$ respectively. 
By construction of PRSs, each state in $\gp$ has at most one parallel state in $\gi$, and marking $\gp$'s existing states can be done in one simultaneous traversal of the two DFAs, using any exploration that covers all the states of $\gi$. 

The new states are a new pattern instance $\phat$ in $\gp$, excluding its initial and possibly its exit state.
The initial state of $\phat$ is the existing state $q_s'\in Q'\setminus Q_N$ that has outgoing new transitions. 
The exit state $q_X'$ of $\phat$ is identified by the following {\em Exit State Discovery} algorithm: 
\begin{enumerate}
    \item If $q_s'$ has incoming new transitions, then $\phat$ is circular: $q_X'=q_s'$. (Fig. \ref{fig:application-of-comp-rule}(i), (ii)).
    
    \item Otherwise $p$ is non-circular.  If  $\phat$ is the first (with respect to the DFA sequence) non-circular pattern to have been grafted onto $q_s'$, then $q_X$ is the unique new state whose transitions into $\dfai{i+1}$ are the \emph{connecting} transitions from Definition \ref{defn:acyclic-rule} (Fig. \ref{fig:application-of-comp-rule} (iii)).
    
    \item If there is no such state then $\phat$ is not the first non-circular pattern grafted onto $q_s'$.  In this case, $q_X'$ is the unique existing state $q_X' \ne q_s'$ with new incoming transitions but no new outgoing transitions.  (Fig. \ref{fig:application-of-comp-rule}(iv)).
\end{enumerate}
Finally, the new pattern instance is $p=\langle\Sigma,q_s',Q_p,q_X',\delta_p \rangle$, where $Q_p=Q_N\cup\{q_s',q_X'\}$ and $\delta_p$ is the restriction of $\delta_N$ to the states of $Q_p$. 

\subsubsection{Discovering the pattern $\phat$} \label{infer-join}
Once we have found the pattern $\phat^3$ in step \ref{InferAlg-step2a}, we need to find the pattern $\phat$ to which it has been grafted. 
We begin with some observations:
\begin{enumerate}
    \item The join state of a composite pattern is always different from its initial and exit states (\emph{edge states}): we cannot compose circular patterns, and there are no `null' patterns.
    \item For every two enabled pattern instances $\phat,\phat'\in \imapi{i}$, $\phat\neq\phat'$, exactly 2 options are possible: either \begin{enumerate*}
    \item every state they share is an edge state of at least one of them, or
    \item one ($p^s$) is contained entirely in the other ($p^c$), and the containing pattern $p^c$ is a composite pattern with join state $q_j$ such that $q_j$ is either one of $p^s$'s edge states, or $q_j$ is not in $p^s$ at all. \end{enumerate*}
\end{enumerate}

Together, these observations imply that
no two enabled pattern instances in a DFA can share a join state. We prove these observations in Appendix \ref{sect:appendixObs}.

Finding the pattern $\phat$ onto which $\phat^3$ has been grafted is now straightforward. Denoting $q_j$ as the parallel of $\phat^3$'s initial state in $\dfai{i}$, we seek the enabled composite
pattern instance $(p,q)\in\imapi{i}$ for which $\joinf(p,q,\dfai{i})=q_j$. If none is present, we seek the only enabled instance $(p,q)\in\imapi{i}$ that contains $q_j$ as a non-edge state, 
but is not yet marked as a composite. (Note that if two enabled instances share a non-edge state, we must already know that the containing one is a composite, otherwise we would not have 
found and enabled the other).

\subsection{Correctness}
\label{Sect:inference-correctness}
\label{Sect:InferenceCorrectness}
\begin{definition}
A PRS $\prs=\langle \Sigma,P,P_c,R\rangle$ is a {\em  minimal generator (MG)} of a sequence of DFAs $S= \dfa_1,\dfa_2,...\dfa_n$ iff it is sufficient and necessary for that sequence, i.e.: \begin{enumerate*}
    \item it generates $S$,
    \item removing any rule $r\in R$ would render $\prs$ insufficient for generating $S$, and
    \item removing any element from $\Sigma,P,P_c$ would make $\prs$ no longer a PRS.
\end{enumerate*}
\end{definition}

\begin{lemma}
Given a finite sequence of DFAs, the minimal generator of that sequence, if it exists, is unique.
\label{lemma:minimal-generator}
\end{lemma}

\begin{theorem}
Let  $\dfa_1,\dfa_2,...\dfa_n$ be a finite sequence of DFAs that has a minimal generator $\prs$.   Then the PRS Inference Algorithm will discover $\prs$.
\label{thm:inference-correctness}
\end{theorem}
The proofs for these claims are given in Appendix \ref{Sect:appendixCorrect}.

\subsection{Deviations from the PRS framework}
\label{Sect:deviations}
Given a sequence of DFAs generated by the rules of PRS $\prs$, the inference algorithm given above will faithfully infer $\prs$ (Section \ref{Sect:InferenceCorrectness}).
In practice however, we will want to apply the algorithm to a sequence of DFAs extracted from a trained RNN using the \lstar algorithm (as in \cite{WeissGoldbergYahav2018}). Such a sequence may contain noise: artifacts from an imperfectly trained RNN, or from the behavior of \lstar (which does not necessarily create PRS-like sequences).
The major deviations are incorrect pattern creation,
simultaneous rule applications, and slow initiation.

\paragraph{Incorrect pattern creation}
\label{Sect:IncorrectPats}

Either due to inaccuracies in the RNN classification, or as artifacts of the \lstar process, incorrect patterns are often inserted into the DFA sequence. Fortunately, the incorrect patterns that get inserted are somewhat random and so rarely repeat, and we can discern between the `legitimate' and `noisy' patterns being added to the DFAs using a {\it voting} and {\it threshold} 
scheme.

The \emph{vote} for each discovered pattern $p\in P$ is the number of times it has been inserted as the new pattern between a pair of DFAs $A_i,A_{i+1}$ in $S$. We set a \emph{threshold} for the minimum vote a pattern needs to be considered valid, and only build rules around the connection of valid patterns onto the join states of other valid patterns. To do this, we modify the flow of the algorithm: before discovering rules, we first filter incorrect patterns.    

We modify step 2 of the algorithm, splitting it into two phases:
\emph{Phase 1:} Mark the inserted patterns between each pair of DFAs, and compute their votes. Add to $P$ those whose vote is above the threshold.
\emph{Phase 2:} Consider each DFA pair $\dfai{i},\dfai{i+1}$ in order. If the new pattern in $\dfai{i+1}$ is valid, and its initial state's parallel state in $\dfai{i}$ also lies in a valid pattern, then synthesize the rule adding that pattern according to the original algorithm in Section \ref{Sect:rule-inference}. Whenever a pattern is discovered to be composite, we add its composing patterns as valid patterns to $P$.

A major obstacle to our research was producing a high quality sequence of DFAs faithful to the target language, as almost every sequence produced has some
noise.  The voting scheme greatly extended the reach of our algorithm.

\paragraph{Simultaneous rule applications} In the theoretical framework, $\gp$ differs from $\gi$ by applying a \emph{single} PRS rule, and therefore $q_s'$ and $q_X'$ are uniquely defined.
\lstar however does not guarantee such minimal increments between DFAs. In particular, it may apply multiple PRS rules between two subsequent DFAs, extending $\dfai{i}$ with several patterns.
To handle this, we expand the initial and exit state discovery methods given in Section \ref{Sect:PatDiscovery}: 
\begin{enumerate}
    \item Mark the new states and transitions $Q_N$ and $\delta_N$ as before.
    \item Identify the \emph{set} of new pattern instance initial states (\emph{pattern heads}): the set $H\subseteq Q'\setminus Q_N$ of states in $\gp$ with outgoing new transitions.
    \item For each pattern head $q'\in H$, compute the \emph{relevant} sets $\delta_{N|q'}\subseteq \delta_N$ and $Q_{N|q'}\subseteq Q_N$ of new transitions and states: the members of $\delta_N$ and $Q_N$ that are reachable from $q'$ \emph{without passing through any existing transitions}.
    \item For each $q'\in H$, restrict to $Q_{N|q'}$ and $\delta_{N|q'}$ and compute $q_X'$ and $p$ as before.
\end{enumerate}
If $\dfai{i+1}$'s new patterns have no overlap and do not create an ambiguity around join states (e.g., do not both connect into instances of a single pattern whose join state has not yet been determined), then they may be handled independently and in arbitrary order. They are used to discover rules and then enabled, as in the original algorithm.

Simultaneous but dependent rule applications -- such as inserting a pattern and then grafting another onto its join state -- are more difficult to handle, as it is not always possible to determine which pattern was grafted onto which.
However, there is a special case which appeared in several of our experiments (examples L13 ad L14 of Section \ref{sect:experiments}) for which we developed a technique as follows:

Suppose we discover a rule $r_1: p_0 \rulearrow_s (p_l \compose p_r) \connect p$, and $p$ contains a cycle $c$ around some internal state $q_j$. If later
 another rule inserts a pattern $p_n$ at the state $q_j$, we understand that $p$ is in fact a composite pattern, with
$p = p_1 \compose p_2$ and join state $q_j$. However, as patterns do not contain cycles at their edge states,  $c$ cannot be a part of either $p_1$ or $p_2$.
We conclude that the addition of $p$ was in fact a simultaneous application of two rules: $r_1':p_0\rulearrow_s (p_l\compose p_r)\connect p'$ and $r_2:p'\rulearrow_c (p_1\compose p_2)\connect c$, where $p'$ is $p$ without the cycle $c$, and update our PRS and our DFAs' enabled pattern instances accordingly.  The case when $p$ is circular is handled similarly.  

\paragraph{Slow initiation}
\label{Sect:first-dfa}
Ideally, $\dfai{1}$ would directly supply an initial rule $\bot \rulearrow p^I$ to our PRS.
In practice, we found that the first couple of DFAs generated by \lstar -- which deal with extremely short sequences -- have completely incorrect structure, and it takes the algorithm some time to stabilise.
Ultimately we solve this by leaving discovery of the initial rules to the \emph{end} of the algorithm, at which point we have a set of `valid' patterns that we are sure are part of the PRS. From there we examine the {\em last} DFA $\dfai{n}$ generated in the sequence, note all the enabled instances $(\pinit,q_0)$ at its initial state, and generate a rule
$\bot \rulearrow \pinit$ for each of them.  
Note however that this technique will not recognise patterns $\pinit$ that do not also appear as an extending pattern $p_3$ elsewhere in the sequence (and therefore do not meet the threshold).

\section{Converting a PRS to a CFG}
\label{Sect:expressibility}
We present an algorithm to convert a given PRS to a context free grammar (CFG), making the rules extracted by our algorithm more accessible. 

\paragraph{A restriction:}
Let $\prs = \langle \Sigma, P, P_c, R\rangle$ be a PRS.   For simplicity, we 
restrict the PRS so that every pattern $p$ can only appear on the LHS of rules of type \ref{rule-composite-circular}
or only on the LHS of rules of type \ref{rule-composite} but cannot only appear on the LHS of both
types of rules.   Similarly, we assume that for each rule $\perp \rightarrow p_I$, the RHS patterns $p_I$ are all
circular or non-circular\footnote{
This restriction is natural: Dyck grammars and all of the examples in Sections \ref{sect:prs-examples} and \ref{Sect:experimental-langs} conform to this restriction. 
}.   In Appendix \ref{Sect:AppendixGeneralizedCfg} we show how to create a CFG without this restriction.

We will create a CFG $G=\langle \Sigma,N,S,Prod \rangle$,
where $\Sigma$, $N$, $S$, and $Prod$ are the terminals (alphabet), non-terminals,
start symbol and productions of the grammar. $\Sigma$ is the same alphabet of $\prs$, and we take $S$ as a special start symbol. We now describe how we obtain $N$ and $Prod$.

For every pattern $p \in P$, let $G_p=\langle \Sigma_p,N_p,Z_p,Prod_p\rangle$ be a CFG describing $L(p)$.  Recall that $P_C$ are composite patterns.
Let $P_Y \subseteq P_C$ be those patterns that appear on the LHS of a rule of type \ref{rule-composite-circular} ($\rulearrow_c$).   
Create the non-terminal $C_{S}$ and for each $p \in P_{Y}$,
create an additional non-terminal $C_p$.  We set 
$N = \{ S, C_{S} \} \bigcup\limits_{p \in P} \{N_p\} \bigcup\limits_{p \in P_Y} \{C_p\}$.

Let $\bot \rulearrow p_I$ be a rule in $\prs$.  If $p_I$ is non-circular, create a production $S ::= Z_{p_I}$.  If 
$p_I$ is circular, create the productions $S ::= S_{C}$, $S_{C} ::= S_{C} S_{C}$ and
$S_{C} ::= Z_{p_I}$.

For each rule $p \rulearrow_s (p_1 \compose p_2) \connect p_3$ create a production $Z_p ::= Z_{p_1} Z_{p_3} Z_{p_2}$.  
For each rule  $p \rulearrow_c (p_1 \compose p_2) \connect p_3$ create the productions 
$Z_p ::= Z_{p_1} C_p Z_{p_2}$, $C_p ::= C_p C_p$, and $C_p ::= Z_{p_3}$. 
Let $Prod'$ be the all the productions defined by the above process.  We set $Prod = \{ \bigcup\limits_{p \in P} Prod_p \} \cup Prod'$. 
\begin{theorem}
Let $G$ be the CFG constructed from $\prs$ by the procedure given above.  Then $L(\prs) = L(G)$.
\label{theorem:CFG}
\end{theorem} 

The proof is given in Appendix \ref{Sect:AppendixCFG}.

\paragraph{The class of languages expressible by a PRS}
Every RE-Dyck language (Section \ref{Sect:Dyck}) can be expressed by a PRS.  But the converse is not true; an RE-Dyck language requires that any delimiter pair can be embedded in any other delimiter pair while a PRS grammar provides more control over which
delimiters can be embedded in which other delimiters.   For instance, the language L12 of Section \ref{Sect:experimental-langs} contains $2$ pairs of
delimiters and only includes strings in which the first delimiter pair is embedded in the second delimiter pair and vice versa.  L12 is expressible by a PRS but is not a Dyck language.  Hence the class of PRS languages are more expressive than Dyck languages and are
contained in the class of CFLs.  But not every CFL can be expressed by a PRS.  See Appendix \ref{Sect:AppendixPRSLimitation}.  

\paragraph{Succinctness} The construction above does not necessarily yield a minimal CFG $G$ equivalent to $\prs$. For a PRS defining the Dyck language of order 2 - which can be expressed by a CFG with 4 productions and one non-terminal - our construction yields a CFG with 10 non-terminals and 12 productions. 
  
In general, the extra productions can be necessary to provide more control over what delimiter pairs can be nested in other delimiter pairs as described above.   
However, when these productions are not necessary, we can often post-process the generated CFG to remove unnecessary productions.  See Appendix \ref{Sect:AppendixCFGexamples}
for the CFGs generated for the Dyck language of order 2 and for the language of alternating delimiters.

\section{Experimental results}
\label{sect:experiments}
\subsection{Methodology}
We test the algorithm on several PRS-expressible context free languages, attempting to extract them from trained RNNs using the process outlined in Figure \ref{fig:overview}.   
For each language, we create a probabilistic CFG generating it, train an RNN on samples from this grammar, extract a sequence of DFAs from the RNN, and apply our PRS inference algorithm\footnote{The implementation needs some expansions to fully apply to complex multi-composed patterns, but is otherwise complete and works on all languages described here.
}. Finally, we convert the extracted PRS back to a CFG, and compare it to our target CFG.

In all of our experiments, we use a vote-threshold s.t. patterns with less than $2$ votes are not used to form any PRS rules (Section \ref{Sect:IncorrectPats}). 
Using no threshold significantly degraded the results by including too much noise, while higher thresholds often caused us to overlook
correct patterns and rules.  

\subsection{Generating a sequence of DFAs}
\label{Sect:DFAgeneration}
We obtain a sequence of DFAs for a given CFG using only positive samples\cite{Gold67,Angluin80} by training a \emph{language-model RNN} (LM-RNN)
on these samples and then extracting DFAs from it with the aid of the \lstar algorithm \cite{Lstar}, as described in \cite{WeissGoldbergYahav2018}. 
To apply \lstar we must treat the LM-RNN as a binary classifier. We set an `acceptance threshold' $t$ and define the RNN's language as the set of sequences $s$ satisfying:
\begin{enumerate*}
\item the RNN's probability for an end-of-sequence token after $s$ is greater than $t$, and
\item at no point during $s$ does the RNN pass through a token with probability $<t$.
\end{enumerate*} This is identical to the concept of \emph{locally $t$-truncated support} defined in \cite{hewitt-etal-2020-rnns}. (Using the LM-RNN's probability for the entire sequence has the flaw that this decreases for longer sequences.)

To create the samples for the RNNs, we write a weighted version of the CFG, in which each non-terminal is given a probability over its rules. We then take $N$ samples from the weighted CFG according to its distribution, split them into train and validation sets, and train an RNN on the train set until the validation loss stops improving. In our experiments, we used $N=10,000$.
For our languages, we used very small 2-layer LSTMs: hidden dimension 10 and input dimension 4.

In some cases, especially when all of the patterns in the rules are several tokens long, the extraction of \cite{WeissGoldbergYahav2018} terminates too soon: neither \lstar nor the RNN abstraction consider long sequences, and equivalence is reached between the \lstar hypothesis and the RNN abstraction despite neither being equivalent to the 'true' language of the RNN. In these cases we push the extraction a little further using two methods: first, if the RNN abstraction contains only a single state, we make an arbitrary initial refinement by splitting 10 hidden dimensions, and restart the extraction. If this is also not enough, we sample the RNN according to its distribution, in the hope of finding a counterexample to return to \lstar. The latter approach is not ideal: sampling the RNN may return very long sequences, effectively increasing the next DFA by many rule applications.

In other cases, the extraction is long, and slows down as the extracted DFAs grow. We place a time limit of $1,000$ seconds ($\sim17$ minutes) on the extraction.

\subsection{Languages}
\label{Sect:experimental-langs}
We experiment on 15 PRS-expressible languages $L_1-L_{15}$, grouped into $3$ classes: 
\begin{enumerate}
    \item Languages of the form {\tt X$^n$Y$^n$}, for various regular expressions {\tt X} and {\tt Y}. In particular, the languages $L_1$ through $L_6$ are {\tt X$_i^n$Y$_i^n$} for: 
    {\tt (X$_1$,Y$_1$)=(a,b)}, \\
    {\tt (X$_2$,Y$_2$)=(a|b,c|d)}, 
    {\tt (X$_3$,Y$_3$)=(ab|cd,ef|gh)},
    {\tt (X$_4$,Y$_4$)=(ab,cd)},\\
    {\tt (X$_5$,Y$_5$)=(abc,def)}, and
    {\tt (X$_6$,Y$_6$)=(ab|c,de|f)}.
    \item Dyck and RE-Dyck languages, excluding the empty sequence. In particular, languages $L_7$ through $L_9$ are the Dyck languages (excluding $\varepsilon$) of order $2$ through $4$, and $L_{10}$ and $L_{11}$ are RE-Dyck languages of order $1$ with the delimiters {\tt (L$_{10}$,R$_{10}$)=(abcde,vwxyz)} and {\tt (L$_{11}$,R$_{11}$)=(ab|c,de|f)}.
    \item Variations of the Dyck languages, again excluding the empty sequence. $L_{12}$ is the language of alternating single-nested delimiters, generating only sequences of the sort {\tt ([([])])} or {\tt [([])]}. $L_{13}$ and $L_{14}$ are Dyck-1 and Dyck-2 with additional neutral tokens {\tt a,b,c} that may appear multiple times anywhere in the sequence. $L_{15}$ is like $L_{13}$ except that the neutral additions are the token {\tt d} and the sequence {\tt abc}, eg: {\tt (abc()())d} is in $L_{15}$, but {\tt a(bc()())d} is not.
\end{enumerate}

\subsection{Results}
\label{Sect:results}
Table \ref{tab:results} shows the results.  
The 2nd column shows the number of DFAs extracted from the RNN. The 
3rd and 4th columns present the number of patterns found by the algorithm before and after applying vote-thresholding to remove noise. The 5th column gives the minimum and maximum votes received by the final patterns\footnote{We count only patterns introduced as a new pattern $p^3$ in some
$\dfai{i+1}$; if $p^3 = p^4 \compose p^5$, but $p^4$ is not introduced independently as a new pattern, we do not count it.}.
The 6th column notes
whether the algorithm found a correct CFG, according to our manual inspection. For languages where our algorithm only missed or included $1$ or $2$ valid/invalid productions, we label it as partially correct.

\begin{table}[t]
\small
    \centering
    \begin{tabular}{c|c|c|c|c|c||c|c|c|c|c|c}
    \toprule 
      LG   &  DFAs &  Init   & Final &  Min/Max  & CFG &  LG   &  DFAs &  Init   & Final &  Min/Max  & CFG \\
             &       &  Pats   & Pats  &  Votes    & Correct &     &       &  Pats   & Pats  &  Votes   & Correct \\
    \midrule
      $L_1$   & 18  & 1 & 1 & 16/16 & Correct                     & $L_9$    & 30 & 6 & 4 & 5/8 & Correct \\ 
      $L_2$   & 16 & 1 & 1 & 14/14 & Correct                        & $L_{10}$   & 6 & 2 & 1 & 3/3 & Correct \\
      $L_3$   & 14 & 6 & 4 & 2/4 & Incorrect                        & $L_{11}$ & 24 & 6 & 3 & 5/12 & Incorrect  \\
      $L_4$   & 8 & 2 & 1 &  5/5 & Correct                          & $L_{12}$   & 28 &  2 & 2 & 13/13 & Correct  \\
      $L_5$ & 10 & 2 & 1 & 7/7 & Correct                            & $L_{13}$   & 9 & 6 & 1 & 2/2 & Correct \\
      $L_6$ & 22 & 9 & 4 & 3/16 & Incorrect                         & $L_{14}$  & 17 & 5 & 2 & 5/7 & Correct \\
      $L_7$   & 24 & 2 & 2 & 11/11 & Correct                        & $L_{15}$   & 13  & 6 & 4 & 3/6 & Incorrect   \\
      $L_8$   & 22 & 5 & 4 & 2/9 & Partial   \\
      \bottomrule
    \end{tabular}
    \caption{Results of experiments on DFAs extracted from RNNs}
    \label{tab:results}
\end{table}

\paragraph{Alternating Patterns} Our algorithm struggled on the languages $L_3$, $L_6$, and $L_{11}$, which contained patterns whose regular expressions had alternations (such as {\tt ab|cd} in $L_3$, and {\tt ab|c} in $L_6$ and $L_{11}$). 
Investigating their DFA sequences uncovered the that the \lstar extraction had `split' the alternating expressions, adding their parts to the DFAs over multiple iterations. For example, in the sequence generated for $L_3$, {\tt ef} appeared in $\dfai{7}$ without {\tt gh} alongside it. The next DFA corrected this mistake but the inference
algorithm could not piece together these two separate steps into a single rule.   
It will be valuable to expand the algorithm to these cases. 

\paragraph{Simultaneous Applications} 
Originally our algorithm failed to accurately generate $L_{13}$ and $L_{14}$ due to simultaneous rule applications.  However, using the technique described in Section \ref{Sect:deviations} we were able
to correctly infer these grammars.  However, more work is needed to handle simultaneous rule applications in general.   

Additionally, sometimes a very large counterexample was returned to \lstar, creating a large increase in the DFAs: the $9$\nth  iteration of the extraction on $L_3$ introduced almost $30$ new states.
The algorithm does not manage to infer anything meaningful from these nested, simultaneous applications.

\paragraph{Missing Rules} For the Dyck languages $L_7-L_9$, the inference algorithm was mostly successful.  
However, due to the large number of possible delimiter combinations, some patterns and nesting relations did not appear often enough in the DFA sequences. As a result, for $L_8$, some productions were missing in the generated grammar.   $L_8$ also created one incorrect
production due to noise in the sequence (one erroneous pattern was generated two times).  When we raised the threshold to require more than $2$ occurrences to be considered a
valid pattern we no longer generated this incorrect production.

\paragraph{RNN Noise}
In $L_{15}$, the extracted DFAs for some reason always forced that a single character {\tt d} be included between every pair of delimiters. Our inference algorithm of course maintained this peculiarity.  It correctly allowed the allowed optional embedding of ``abc'' strings.  But due to noisy (incorrect) generated DFAs, the 
patterns generated did not maintain balanced parenthesis.

\section{Related work}
\label{Sect:related}

\paragraph{Training RNNs to recognize Dyck Grammars.}
Recently there has been a surge of interest in whether RNNs can learn Dyck languages 
\cite{Bernardy2018,sennhauser-berwick-2018-evaluating,skachkova-etal-2018-closing,YuVukuhn2019}.
While these works report very good results on learning the language for sentences of similar distance and depth as the training set, with the exception of \cite{skachkova-etal-2018-closing}, they report significantly less accuracy 
for out-of-sample sentences.

Sennhauser and Berwick \cite{sennhauser-berwick-2018-evaluating} use LSTMs, and show that in order to keep the error rate with a $5$ percent tolerance, the number of hidden units
must grow exponentially with the distance or depth of the sequences\footnote{However, see \cite{hewitt-etal-2020-rnns} where they show that a Dyck grammar
of $k$ pairs of delimiters that generates sentences of maximum depth $m$, only $3m\ceil{\log{k}}-m$ hidden memory units are required, and show experimental results
that confirm this theoretical bound in practice.}.  They also found that out-of-sample results were not very
good. 
They conclude that LSTMs
cannot learn rules, but rather use statistical approximation.   Bernardy \cite{Bernardy2018} experimented with various RNN architectures. 
When they test their RNNs on strings that are at most double in length of the training set, they found that for out-of-sample strings, the accuracy  
varies from about $60$ to above $90$ percent.
The fact that the LSTM has more difficulty in predicting closing
delimiters in the middle of a sentence than at the end leads Bernardy to conjecture that for closing parenthesis the RNN is using a counting mechanism, but has not truly
learnt the Dyck language (its CFG).  
Skachkova, Trost and Klakow \cite{skachkova-etal-2018-closing} experiment with Ellman-RNN, GRU and LSTM architectures.   They provide a mathematical model for the probability of
a particular symbol in the $i^{th}$ position of a Dyck sentence. 
They experiment with how well the models predict the closing delimiter, which they find varying results
per architecture.  However, for LSTMs, they find nearly perfect accuracy across words with large distances and embedded depth.  
    
Yu, Vu and Kuhn \cite{YuVukuhn2019} compares the three works above and 
argue that the task of predicting a closing bracket of a balanced Dyck word, as performed in \cite{skachkova-etal-2018-closing}, is a poor test for
checking if the RNN learnt the language, as it can be simply computed by a counter.  In contrast, their carefully constructed experiments give
a prefix of a Dyck word and train the RNN to predict the next valid closing bracket. 
They experiment with an LSTM using 4 different models, and show
that the generator-attention model \cite{LuongPM2015} performs the best, and is able to generalize quite well at the tagging task . 
However, when using RNNs to complete the entire Dyck word,
while the generator-attention model does quite well with in-domain tests, it degrades rapidly with out-of-domain tests.
They also conclude that RNNs do not really learn the CFG underlying the Dyck language.  These experimental results are reinforced by the theoretical work
in \cite{hewitt-etal-2020-rnns}.  They remark that no finite precision RNN can learn a Dyck language of unbounded depth, and give precise bounds on 
the memory required to learn a Dyck language of bounded depth.   

In contrast to these works, our research tries to extract the CFG from the RNN.   We discover these rules based upon DFAs synthesized from the RNN
using the algorithm in \cite{WeissGoldbergYahav2018}.  Because we can use a short sequence
of DFAs to extract the rules, and because the first DFAs in the sequence describe Dyck words with increasing but limited distance and depth, we are able to extract the CFG perfectly,
even when the RNN does not generalize well.   Moreover, we show that our approach generalizes to more complex types of delimiters, and to Dyck languages with expressions between delimiters.

\paragraph{Extracting DFAs from RNNs.}
There have been many approaches to extract higher level representations from a neural network (NN) 
to facilitate comprehension and verify correctness.   One of the oldest approaches is to extract
rules from a NN \cite{Thrun95extractingrules,Hailesilassie2016}. In order to model state, there have been various approaches to extract FSA from RNNs \cite{omlin-giles-96,Jacobsson2005,wang2019}.   
We base our work on \cite{WeissGoldbergYahav2018}. Its ability to generate sequences of DFAs that increasingly better approximate the CFL is critical to our
method.   

Unlike DFA extraction, there has been relatively little research on extracting a CFG from an RNN.   One exception is \cite{Sun1997}, 
where they develop a Neural Network Pushdown Automata (NNPDA) framework, a hybrid system augmenting an RNN with external stack memory.  The RNN also reads the top of the stack as added input, optionally pushes to or pops the stack after each new input symbol.  
They show how to extract a Push-down Automaton from a NNPDA, however, their technique relies on the PDA-like structure of the inspected architecture.
In contrast, 
we extract CFGs from RNNs 
without stack augmentation.

\paragraph{Learning CFGs from samples.}
There is a wide body of work on learning CFGs from samples.  An overview is given in \cite{D'Ulizia2011} and a survey of work for grammatical inference applied to 
software engineering tasks can be found in \cite{Stevenson2014}.  

Clark et. al. studies algorithms for learning CFLs given only positive examples
\cite{Gold67}.   In \cite{Clark2007Polynomial},
Clark and Eyraud show how one can learn a subclass of CFLs called {\em CF substitutable} languages.
There are many languages that can be expressed by a PRS but are not substitutable, such as $x^nb^n$.  However, there are also substitutable languages that cannot be expressed by a PRS ($wxw^{R}$ - see Appendix \ref{Sect:AppendixPRSLimitation}).
In \cite{Clark2008APA}, Clark, Eyraud and Habrard present Contextual Binary Feature Grammars.
However, it does not include Dyck languages of arbitrary order.   None of these techniques deal with noise in the data, essential to learning a language from an RNN.  While we have focused on practical learning of CFLs, theoretical limits on learning based upon positive examples is well known; see\cite{Gold67,Angluin80}.

\section{Future Directions}
\label{Sect:future}
Currently, for each experiment, we train the RNN on that language and then apply the PRS inference algorithm
on a single DFA sequence generated from that RNN.   Perhaps the most substantial improvement we can make is to extend our technique to learn from multiple DFA sequences.
We can train multiple RNNs (each one based upon a different architecture if desired) and generate 
DFA sequences for each one.  We can then run the PRS inference algorithm on each of these sequences, and generate a CFG
based upon rules that are found in a significant number of the runs.  This would require care to guarantee that the final rules form a cohesive CFG.  It would also address the issue that not
all rules are expressed in a single DFA sequence, and that some grammars may have rules that are executed only once per word of the language.

Our work generates CFGs for generalized Dyck languages, but it is possible to generalize PRSs  to express a greater range of languages.   Work will be needed to extend the PRS inference algorithm to reconstruct grammars
for all context-free and perhaps even some context-sensitive languages.

\appendix
\section{Observation on PRS-Generated Sequences} \label{sect:appendixObs}

We present and prove an observation on PRS-generated sequences used for deriving the PRS-inference algorithm (Section \ref{Sect:PatDiscovery}).

\begin{lemma}
Let $\explicitpairi{i}$ be a PRS-generated \pairname. Then for every two enabled pattern instances $\phat,\phat'\in\imapi{i}, \phat\neq\phat'$, exactly 2 options are possible: \begin{enumerate*}
    \item every state they share is the initial or exit state (\emph{edge state}) of at least one of them, or
    \item one ($\phat^s$) is contained entirely in the other ($\phat^c$), and $\phat^c$ is a composite pattern with join state $q_j$ such that either $q_j$ is one of $\phat^s$'s edge states, or $q_j$ is not in $\phat^s$ at all.
\end{enumerate*}
\end{lemma}

\begin{proof}
We prove by induction. 
For $\explicitpairi{1}$, $|\imapi{1}|\leq 1$ and the lemma holds vacuously. We now assume it is true for $\explicitpairi{i}$.

Applying a rule of type \ref{rule-start} adds only one new instance $\phat^I$ to $\imapi{i+1}$, which shares only its initial state with the existing patterns, and so option $1$ holds. 

Rules of type \ref{rule-composite-circular} and \ref{rule-composite-serial} add up to three new enabled instances, $\phat^1,\phat^2$, and $\phat^3$, to $\imapi{i+1}$. $\phat^3$ only shares its edge states with $\dfai{i}$, and so option (1) holds between $\phat^3$ and all existing instances $\phat'\in\imapi{i}$, as well as the new ones $\phat^1$ and $\phat^2$ if they are added (as their states are already contained in $\dfai{i}$). 

 We now consider the case where $\phat^1$ and $\phat^2$ are also newly added (i.e. $\phat^1,\phat^2\notin\imapi{i}$). We consider a pair $\phat^i,\phat'$ where $i\in\{1,2\}$. As $\phat^1$ and $\phat^2$ only share their join states with each other, and both are completely contained in $\phat$ such that $\phat$'s join state is one of their edge states, the lemma holds for each of $\phat'\in\{\phat^1,\phat^2,\phat\}$. We move to $\phat'\neq \phat^1,\phat^2,\phat$. Note that ($i$) $\phat'$ cannot be contained in $\phat$, as we are only now splitting $\phat$ into its composing instances, and ($ii$), if $\phat$ shares any of its edge states with $\phat^i$, then it must also be an edge state of $\phat^i$ (by construction of composition).

As $\phat^i$ is contained in $\phat$, the only states that can be shared by $\phat^i$ and $\phat'$ are those shared by $\phat$ and $\phat'$.
If $\phat,\phat'$ satisfy option 1, i.e., they only share edge states, then this means any states shared by $\phat'$ and $\phat^i$ are edge states of $\phat'$ or $\phat$. Clearly, $\phat'$ edge states continue to be $\phat'$ edge states. As for each of $\phat$'s edge states, by ($ii$), it is either not in $\phat^i$, or necessarily an edge state of $\phat^i$. Hence, if $\phat,\phat'$ satisfy option $1$, then $\phat^i,\phat'$ do too. 

Otherwise, by the assumption on $\explicitpairi{i}$, option $2$ holds between $\phat'$ and $\phat$, and from ($i$) $\phat'$ is the containing instance. As $\phat^i$ composes $\phat$, then $\phat'$ also contains $\phat^i$. Moreover, by definition of option $2$, the join state of $\phat'$ is either one of $\phat$'s edge states or not in $\phat$ at all, and so from $(ii)$ the same holds for $\phat^i$. 
\end{proof}

\section{Correctness of the Inference Algorithm}
\label{Sect:appendixCorrect}

\newtheorem*{lemma:minimal-generator}{Lemma \ref{lemma:minimal-generator}}
\begin{lemma:minimal-generator}
Given a finite sequence of DFAs, the minimal generator of that sequence, if it exists, is unique.
\end{lemma:minimal-generator}

{\bf Proof:}
Say that there exists two MGs, $\prs_1 = \langle \Sigma^1, P^1, P_c^1, R^1\rangle$ and $\prs_2 =\langle \Sigma^2, P^2, P_c^2, R^2\rangle$ that generate the sequence $\dfai{1},\dfai{2},\cdots,\dfai{n}$.  
Certainly $\Sigma^1 =\Sigma^2=\bigcup_{i\in[n]}\Sigma^{A_i}$.   

We show that $R^1 = R^2$.
Say that the first time MG1 and MG2 differ from one another is in explaining which rule is used when expanding from
$\dfai{i}$ to $\dfai{i+1}$.   Since MG1 and MG2 agree on all rules used to expand the sequence prior to $\dfai{i+1}$, they agree on the set of patterns enabled in $\dfai{i}$. 
If this expansion is adding a pattern $p_3$ originating at the
initial state of the DFA, then it can only be explained by a single rule $\bot \rulearrow p_3$, and so the explanation of MG1 and MG2 is identical.
Hence the expansion must be created by a rule of type \ref{rule-composite-circular} or \ref{rule-composite-serial}.   
Since the newly added pattern instance $\phat^3$ is is uniquely identifiable in $\dfai{i+1}$, 
$\prs_1$ and $\prs_2$ must agree on the pattern $p^3$ that appears on the RHS of the rule explaining this expansion.  $\phat^3$ is 
inserted at some state $q_j$ of $\dfai{i}$.  $q_j$ must be  the join state of an enabled pattern instance $\phat$ in $\dfai{i}$. But this join state uniquely identifies that pattern: as noted in Section \ref{Sect:rule-inference}, no two enabled patterns in a \pairname \ share a join state.
Hence $\prs_1$ and $\prs_2$ must agree that the pattern $p =  p^1 \compose p^2$
is the LHS of the rule, and they therefore agree that the rule is $p \rulearrow_s (p^1 \compose p^2) \connect p^3$, if $p^3$ is non-circular, or
 $p \rulearrow_c (p^1 \odot p^2) \connect p^3$ if $p_3$ is circular.  Hence  $R^1 = R^2$.
 
Since $\prs_1$ ($\prs_2$) is an MG, it must be that $p \in P^1$ ($p \in P^2$) iff $p$ appears in a rule in $R^1$ ($R^2$).  Since $R^1 = R^2$, $P^1=P^2$.  
Furthermore, a pattern $p \in P_c$ iff it appears on the LHS of a rule.  Therefore $P_c^1 = P_c^2$. \qed

\newtheorem*{thm:inference-correctness}{Theorem \ref{thm:inference-correctness}}
\begin{thm:inference-correctness}
Let  $\dfa_1,\dfa_2,...\dfa_n$ be a finite sequence of DFAs that has a minimal generator $\prs$.   Then the PRS Inference Algorithm will discover $\prs$.
\end{thm:inference-correctness}

{\bf Proof:}
This proof mimics the proof in the Lemma above.  In this case $\prs_1 =\langle \Sigma^1, P^1, P_c^1, R^1\rangle$ is the MG for this sequence and
$\prs_2 =\langle \Sigma^2, P^2, P_c^2, R^2\rangle$ is the PRS discovered by the PRS inference algorithm.  

We need to show that the PRS inference algorithm faithfully follows the steps above for $\prs_2$.  This straightforward by comparing the steps of the inference algorithm
with the steps for $\prs_2$.   One subtlety is to show that the PRS inference algorithm correctly identifies the new pattern $\phat^3$ in $\dfai{i+1}$ 
extending $\dfai{i}$.  The algorithm easily finds all the newly inserted states and transitions in
$\dfai{i+1}$.   All of the states, together with the initial state, must belong to the new pattern.  However not all transitions necessarily belong to the pattern.  
The Exit State Discovery algorithm of Section
\ref{Sect:PatDiscovery} correctly differentiates between new transitions that are part of the inserted pattern and those that are 
{\it connecting transitions} (The set $C$ of Definition \ref{defn:acyclic-rule}).   Hence the algorithm correctly finds the new pattern in $\dfai{i+1}$. \qed

\section{The expressibility of a PRS}
\label{Sect:AppendixCFG}
We present a proof to Theorem \ref{theorem:CFG} showing that the CFG created from a PRS expresses the same language.   

\newtheorem*{thm:CFG}{Theorem \ref{theorem:CFG}}
\begin{thm:CFG}
Let $G$ be the CFG constructed from $\prs$ by the procedure given in Section \ref{Sect:expressibility}.  Then $L(\prs) = L(G)$.
\end{thm:CFG}

{\bf Proof:}
Let $s \in L(\prs)$.  Then there exists a sequence of DFAs $\dfai{1} \cdots \dfai{m}$ generated by $\prs$ s.t. $s \in L(\dfai{m})$.  We will show that $s\in L(G)$. 
W.l.g. we assume that each DFA in the sequence is necessary; i.e., if the rule application to $\dfai{i}$ creating $\dfai{i+1}$ were absent, then $s \notin L(\dfai{m})$.  
We will use the notation $\phat$ to refer to a specific instance of a pattern $p$ in $\dfai{i}$ for some $i$ ($1 \le i \le m$),
and we adopt from Section \ref{sect:PRS} the notion of \emph{enabled} pattern instances.   So, for instance, if we 
apply a rule $p \rulearrow_s (p^1 \compose p^2) \connect p^3$, where $p = p_1 \compose p_2$, to an instance of $\phat$ in 
$\dfai{i}$, then $\dfai{i+1}$ will contain a new path through the enabled pattern instances $\phati{1},\phati{2}$ and $\phati{3}$.

A {\em p-path} (short for {\em pattern-path}) through a DFA $\dfai{i}$ is a path 
$\rho = q_0 \rightarrow^{p_1} q_1 \rightarrow^{p_2} \cdots q_{t-1} \rightarrow^{p_t} q_t$, where $q_0$ and
$q_t$ are the initial and final states of $\dfai{i}$ respectively, and for each transition $q_j \rightarrow^{p_{j+1}} q_{j+1}$,
$q_j$ ($0 \le j \le t-1$) is the initial state of an enabled pattern instance of type $p_{j+1}$ and $q_{j+1}$ is the final state of that pattern instance.
A state may appear multiple times in the path if there is a cycle in the DFA and that state is traversed multiple times.  
If $\phat$ is an enabled circular pattern and the path contains a cycle that traverses that instance of $p$, and only that instance, multiple times consecutively, it is only represented
once in the path, since that cycle is completely contained within that pattern; a p-path cannot contain consecutive self-loops 
$q_j \rightarrow^{p} q_{j} \rightarrow^{p} q_{j}$.  
$Pats(\rho) = p_{1} p_{2} \cdots p_{t}$, the instances of the patterns traversed along the path $\rho$.

We say that a p-path 
$\rho = q_0 \rightarrow^{p_1} q_1 \rightarrow^{p_2} \cdots q_{t-1} \rightarrow^{p_t} q_t$ through $\dfai{m}$
is an {\em acceptor} (of $s$) iff $s = s_1 \cdots s_t$ and $s_i \in L(p_i)$ for all $i$ ($1 \le i \le t$).
DFAs earlier in the sequence are not acceptors as they contain patterns that have not yet been expanded.   But
we can ``project'' the final p-path onto a p-path in an earlier DFA.  We do so with the following definition of
a {\em p-cover}:
\begin{itemize}
    \item If a path $\rho$ is an acceptor, then it is a p-cover.
    
    \item Let $p$ be a pattern and let $\dfai{i+1}$ be obtained from $\dfai{i}$ by application of the rule $p \rulearrow_s (p^1 \compose p^2) \connect p^3$ or
    $p \rulearrow_c (p^1 \odot p^2) \connect p^3$ to 
    $\phat$ in $\dfai{i}$ obtaining a sub-path $q_1 \rightarrow^{p_1} q_3 \rightarrow^{p_3} q_4 \rightarrow^{p_2} q_2$ through instances $\phati{1},\phati{2}$ and $\phati{3}$.  Furthermore, say that the p-path $\rhoi{i+1}$ through $\dfai{i+1}$ is a p-cover.
    Then the path $\rhoi{i}$ through $\dfai{i}$ is p-cover, where $\rhoi{i}$ is obtained from $\rhoi{i+1}$ by replacing each occurrence
    of $q_1 \rightarrow^{p_1} q_3 \rightarrow^{p_3} q_4 \rightarrow^{p_2} q_2$ in $\rhoi{i+1}$ traversing $\phati{1},\phati{3}$ and $\phati{2}$ by the single transition 
    $q_1 \rightarrow^p q_2$ traversing $\phat$ in $\rhoi{i}$.  (If $p$ is circular then $q_1 = q_2$).  
    If this results in consecutive self loops $q_1 \rightarrow^{p} q_1 \rightarrow^{p} q_1$ we collapse them into a single cycle, $q_1 \rightarrow^{p} q_1$.
   
   \item Let $\dfai{i+1}$ be obtained by applying a rule $\perp \rulearrow \pinit$ to $\dfai{i}$ obtaining 
   an instance of $\phat^{I}$, where $\pinit$ is a circular pattern (Defn. \ref{defn:start-rule-existing-dfa}).
   Furthermore, say that the p-path $\rhoi{i+1}$ through $\dfai{i+1}$ is a p-cover.   Then the path $\rhoi{i}$ through $\dfai{i}$ is 
   p-cover, where $\rhoi{i}$ is obtained from $\rhoi{i+1}$ by replacing each occurrence of $q_0 \rightarrow^{\pinit} q_0$
   traversing $\phat^{I}$ by the single state $q_0$.   
  
\end{itemize}
Hence we can associate with each $\dfai{i}, 1 \le i \le m$ a unique p-cover $\rhoi{i}$.

Let $\tree$ be a partial derivation tree for the CFG $G$, where every branch of the tree terminates with a non-terminal
$Z_{p}$ for some pattern $p$.  We write $\hat{Z_{p}}$ for a particular instance of $Z_{p}$ in $\tree$. 
$Leaves(\tree)$ is the list of patterns obtained by concatenating all the leaves 
(left-to-right) in $\tree$ and replacing each leaf $Z_{p_k}$ by the pattern $p_k$.

We claim that for each $\dfai{i}$ with p-cover $\rhoi{i}$ there exists a partial derivation tree $\treei{i}$ such that 
$Pats(\rhoi{i}) = Leaves(\treei{i})$.   We show this by induction.  

For the base case, consider $\dfai{1}$, which is formed by application of a rule $\bot \rulearrow p^I$.   By construction
of $G$, there exists a production $S ::= Z_{p^I}$.  
$\rhoi{1} = s_0 \rightarrow^{p^I} s_f$, where $S_0$ and $s_f$ are the initial and final states of $\pinit$ respectively, and let $\treei{1}$ be the tree formed by application of the production $S ::= Z_{p^I}$
creating the instance $\hat{Z_{p^I}}$.   Hence $Pats(\rhoi{1}) = \pinit = Leaves(\treei{1})$. 

For the inductive step assume that for $\dfai{i}$ there exists $\treei{i}$ s.t. $Pats(\rhoi{i}) = Leaves(\treei{i})$.
Say that $\dfai{i+1}$ is formed from $\dfai{i}$ by applying the rule
$p \rulearrow_c (p^1 \odot p^2) \connect p^3$ (of type \ref{rule-composite-circular})
or $p \rulearrow_s (p^1 \compose p^2) \connect p^3$ (of type \ref{rule-composite-serial})  
to an instance $\phat$ of $p$ in $\dfai{i}$,
where the initial state of $\phat$ is $q_1$ and its final state is $q_2$ ($q_1 = q_2$ if $p$ is circular) and there
is a sub-path in $\dfai{i}$ of the form $q_1 \rightarrow^p q_2$. 
After applying this rule there is
an additional sub-path $q_1 \rightarrow^{p_1} q_3 \rightarrow^{p_3} q_4 \rightarrow^{p_2} q_2$ in $\dfai{i+1}$ traversing $\phati{1},\phati{3}$ and $\phati{2}$.  
We consider two cases:

Case 1.  $p$ is non-circular.  
The sub-path $q_1 \rightarrow^p q_2$ may appear multiple times in $\rhoi{i}$ even though $p$ is non-circular, since it may be part of a larger cycle.
Consider one of these instances where $q_1 \rightarrow^p q_2$  gets replaced by $q_1 \rightarrow^{p_1} q_3 \rightarrow^{p_3} q_4 \rightarrow^{p_2} q_2$ in $\rhoi{i+1}$.  
Say that this instance of $\phat$ is represented by pattern $p$ at position $u$ in $Pats(\rhoi{i})$. In $\rhoi{i+1}$, the sub-list of patterns $p_1,p_3,p_2$ will replace $p$ at that position (position $u$).
By induction there is a pattern $p$ in $Leaves(\treei{i})$ at position $u$ and let $\zphat$ be the non-terminal instance
in $\treei{i}$ corresponding to that pattern $p$.   If the rule being applied is of type \ref{rule-composite-serial} then, by construction of $G$, there exists a production $Z_p ::= Z_{p_1} Z_{p_3} Z_{p_2}$.  
We produce $\treei{i+1}$ by extending $\treei{i}$ at that instance of $Z_p$ by applying that production to 
$\zphat$.  If the rule is of type \ref{rule-composite-circular}, then we produce $\treei{i+1}$ by extending $\treei{i}$ 
at that instance of $Z_p$ by applying the productions $Z_p ::= Z_{p_1} C_p Z_{p_2}$ and $C_p ::= Z_{p_3}$, which
exist by the construction of $G$.   Hence both $Pats(\rhoi{i+1})$ and $Leaves(\treei{i+1})$ 
will replace $p$ at position $u$ by $p_1,p_3,p_2$.  We do this for each traversal of $\phat$ in $\rhoi{i}$ that gets 
replaced in $\rhoi{i+1}$ by the traversal of $\phati{1}, \phati{3}$, and $\phati{2}$.  
By doing so, $Pats(\rhoi{i+1}) = Leaves(\treei{i+1})$.

Case 2: $p$ is circular.  
This is similar to the previous case except this time, since $p$ is circular, we may need to replace a single sub-path
$q_1 \rightarrow^{p} q_{1}$ corresponding to an instance of $\phat$ in $\rhoi{i}$ by multiple explicit cycles
as defined by $\rhoi{i+1}$.  Each cycle will either traverse $q_1 \rightarrow^{p} q_{1}$ 
or the longer sub-path $q_1 \rightarrow^{p_1} q_3 \rightarrow^{p_3} q_4 \rightarrow^{p_2} q_1$.
 
Say that there exists an instance $\phat$ represented by pattern $p$ at position $u$ in $Pats(\rhoi{i})$
that gets replaced in $\rhoi{i+1}$ by explicit cycles; i.e., $\rhoi{i+1}$ 
replaces $q_1 \rightarrow^{p} q_{1}$ traversing $\phat$ in $\rhoi{i}$ with a new sub-path $\sigma$ in $\rhoi{i+1}$ containing $x$ cycles $q_1 \rightarrow^{p_1} q_3 \rightarrow^{p_3} q_4 \rightarrow^{p_2} q_1$ interspersed with $y$ cycles $q_1 \rightarrow^{p} q_{1}$, where $p= p_1 \compose_c p_2$.  (Per definition of a p-path, there cannot be two consecutive instances of these latter cycles).
Hence in total $\sigma$ may enter and leave $q_1$ a total of $z = x + y$ times. 
By induction there is a pattern $p$ in $Leaves(\treei{i})$ at position $u$ and let $\zphat$ be the non-terminal instance
in $\treei{i}$ corresponding to that pattern $p$.   By construction of $G$, since $p$ is circular, the parent of 
$\zphat$ is an instance $\hat{C_{p'}}$ of the non-terminal $C_{p'}$ for some pattern $p'$ and there exists productions 
$C_{p'} ::= C_{p'} C_{p'}$, and $C_{p'} ::= Z_{p}$.  Using these productions we replace this single instance $\hat{C_{p'}}$ by $z$ copies of $C_{p'}$.  
If the $j^{th}$ cycle of $\sigma$ is $q_1 \rightarrow^{p} q_{1}$ then we have the 
$j^{th}$ instance of $C_{p'}$ derive $Z_p$ without any further derivations.   If the $j^{th}$ cycle is $q_1 \rightarrow^{p_1} q_3 \rightarrow^{p_3} q_4 \rightarrow^{p_2} q_1$, then we also have the 
$j^{th}$ instance of $C_{p'}$ derive $Z_p$.   However, if the rule being applied is of type \ref{rule-composite-serial} then that instance of $Z_p$ derives $Z_{p_1} Z_{p_3} Z_{p_2}$.  
If it is of type \ref{rule-composite-circular} then that instance of $Z_p$ derives $Z_{p_1} C_{p} Z_{p_2}$ and $C_p$ derives $Z_{p_3}$.  
Hence both $Pats(\rhoi{i})$ and $Leaves(\treei{i})$ 
will replace $p$ at position $u$ by $x$ copies of $p_1,p_3,p_2$ intermixed with $y$ copies of $p$.  We do this for each traversal of $\phat$ in $\rhoi{i}$ that gets expanded in $\rhoi{i+1}$ by application of this rule.    
By doing so, $Pats(\rhoi{i+1}) = Leaves(\treei{i+1})$.  

To complete the inductive step, we need to consider the case when $\dfai{i+1}$ is formed from $\dfai{i}$
by applying a rule $\perp \rulearrow p^{I'}$, where $p^{I'}$ is circular, per Defn. \ref{defn:start-rule-existing-dfa}. 
This will insert
$p^{I'}$ into $Pats(\rhoi{i+1})$ at a point when $\rhoi{i}$ is at the initial state $q_0$.  Say that there exists
a sub-path $\sigma = q_0 \rightarrow^{p_1} q_1 \rightarrow^{p_2} \cdots q_e \rightarrow^{p_e} q_0$ in $\rhoi{i}$.   Then 
the application of this rule may add the sub-path $q_0 \rightarrow^{p^{I'}} q_0$ either at the beginning or end of
$\sigma$ in $\rhoi{i+1}$.  W.l.g. assume it gets asserted at the end of this sub-path, and $p_e$ 
occurs at position $u$. Then $Pats(\rhoi{i+1})$ will extend $Pats(\rhoi{i})$
by inserting $p^{I'}$ at position $u+1$ in $\rhoi{i}$.   Since $\sigma$ is a cycle, starting and ending at $q_0$, there must be an instance $\hat{C_S}$ of 
$C_S$ in $\treei{i}$ where $C_S$ is derived by one or more productions of the form
$S::=C_S$ and $C_S ::= C_S \ C_S$. Furthermore, $\hat{C_S}$ derives a sub-tree $T$ s.t. $Leaves(T) = Pats(\sigma)$.  
By construction of $G$, there exists a production $C_S ::= C_{p_{I}'}$.  
We add the production $C_S ::= C_S \ C_S$ to $\hat{C_S}$ so that the first child $C_S$ derives $T$ as in $\treei{i}$.
At the second instance we apply the production $C_S ::= C_{p_{I}'}$.  Hence $p_{I}'$ will appear at position $u+1$ in
$\treei{i+1}$.  We repeat this for each cycle involving $q_0$ in $\rhoi{i}$ that gets extended by the pattern $p^{I'}$ in 
$\rhoi{i+1}$.  By doing so, $Pats(\rhoi{i+1}) = Leaves(\treei{i+1})$.  A similar argument holds if $p^{I'}$ is added
to the first position in $Pats(\rhoi{i+1})$.

Hence we have shown that  $Pats(\rhoi{m}) = Leaves(\treei{m})$. Let $Pats(\rhoi{m}) = p_1 \cdots p_t$.
Since $\rhoi{m}$ is an acceptor for $s$, it must be that there exists $s_j \in \Sigma^+$ ($1 \le j \le t$) s.t. 
$s_j \in L(p_j)$ and $s = s_1 \cdots s_t$.   But since $Leaves(\treei{m}) = Z_{p_1} \cdots Z_{p_t}$ and each
$Z_{p_j}$ can derive $s_j$, we can complete the derivation of $\treei{m}$ to derive $s$.
This shows that $s \in L(PR) \implies s \in L(G)$.  The converse is also true and can be shown by similar technique so
we leave the proof to the reader. \qed

\subsection{Constructing a CFG from an unrestricted PRS}
\label{Sect:AppendixGeneralizedCfg}
The construction of Section \ref{Sect:expressibility} assumed a restriction that a pattern $p$ cannot appear on the LHS
of rules of type \ref{rule-composite-circular} and of type \ref{rule-composite-serial}.   I.e., we cannot have two rules of the form
$p \rulearrow_c (p^1 \odot p^2) \connect {p'}^3$ and $p \rulearrow_s (p^1 \compose p^2) \connect p^3$.
If we were to allow both of these rules then one could construct a path through a DFA instance that first
traverses an instance of $p^1$, then traverses instance of the circular pattern ${p'}^3$ any number of times, then traverses an
instance of $p^3$, and then traverses $p^2$.   However the current grammar does not allow such constructions;   
the non-terminal $Z_p$  can  {\em either} derive $Z_{p_1}$ followed by $Z_{p_3}$
followed by $Z_{p_2}$ or, in place of $Z_{p_3}$, any number of instances of $C_p$ that in turn derives $Z_{p_3'}$.

Hence to remove this restriction, we modify the constructed CFG.  Following Section \ref{Sect:expressibility}, for
every pattern $p \in P$, $G_p$ is the CFG with Start symbol $Z_p$ and non-terminals $N_p$.
$P_Y$ are the patterns appearing on the LHS of some rule of type \ref{rule-composite-circular}.
Given the PRS $\prs = \langle\Sigma, P, P_C, R\rangle$ we create a CFG $G=(\Sigma,N,S,Prod)$, where 
$N = \{ S, C_{S}, E_{S}\} \bigcup\limits_{p \in P} \{N_p, E_p\} \bigcup\limits_{p \in P_Y} \{C_p\}$.

Create the productions $S ::= E_S$, $S::= C_S E_S$ and $C_{S} ::= C_{S} C_{S}$.   Let $\bot \rulearrow p^I$ be a rule in $\prs$.  
Create the production $E_S ::= Z_{p^I}$.  If $p^I$ is circular, create the additional production $C_{S} ::= Z_{p^I}$.

For each rule $p \rulearrow_c (p^1 \odot p^2) \connect p^3$ or
$p \rulearrow_s (p^1 \compose p^2) \connect p^3$ create the productions
$Z_p \rulearrow Z_{p_1} E_{p} Z_{p_2}$ and $E_{p} ::= Z_{p_3}$.  
For each rule  $p \rulearrow_c (p^1 \odot p^2) \connect p^3$ create the additional productions 
$Z_p ::= Z_{p_1} C_p E_{p} Z_{p_2}$, $C_p ::= C_p C_p$, and $C_p ::= Z_{p_3}$. 
Let $Prod'$ be the all the productions defined by the process just given.  $Prod = \{ \bigcup\limits_{p \in P} Prod_p \} \cup Prod'$.   

\subsection{Example of a CFG generated from a PRS}
\label{Sect:AppendixCFGexamples}
The following is the CFG generated for the Dyck Language of order 2 ($L_7$ of Section \ref{Sect:experimental-langs})\footnote{As previously noted,
the terminals, in this case ``(",``)", ``[",``]", are actually represented as base patterns.}.
\begin{verbatim}
S ::= SC
SC ::= SC SC | P1 | P2
P1::= ( P1C )
P1C ::= P1C  P1C | P1 | P2
P2::= [ P2C ]
P2C ::= P2C  P2C | P1 | P2
\end{verbatim}

As remarked in Section \ref{Sect:expressibility}, the usual CFG for Dyck languages collapses all of these non-terminals into a single one.
However, sometimes the extra non-terminals generated by the algorithm are necessary, as illustrated by the following CFG for alternating delimiters ($L_{12}$ of Section \ref{Sect:experimental-langs}) generated by the algorithm.
\begin{verbatim}
S ::= P1 | P2
P1::= ( P2 )
P2::= [ P1 ]
\end{verbatim}

\subsection{Limitations on the expressibility of a PRS}
\label{Sect:AppendixPRSLimitation}
Not every CFL is expressible by a PRS. Consider the language $L_{axb} = \{a^ixb^i : i \ge 0\}$.   
Assume there exists a PRS $\prs$ s.t. $L(\prs) = L_{axb}$.  $\prs$ contains a finite number of initial rules, and each 
initial rule must contain a finite number of straight (non-branching) paths from the start state $q_0$ to the final state $q_f$ (otherwise, if two paths recognizing two distinct strings shared a state, 
it is easy to see it would accept a string not in $L_{axb}$).   Therefore there must be at least one pattern $p^I$ s.t.
there exists a rule $\bot \rightarrow p^I$ and $p^I$ appears on the LHS of some rule.   In particular, there is a path $\alpha$
in $p^I$ from $q_0$ to $q_f$ that recognizes $a^jxb^j$ for some $j>0$.  Let $\alpha = \alpha_1 \alpha_2$ with the join state 
in between $\alpha_1$ and $\alpha_2$.   After applying this rule to the join state of $p^I$, a new path is created of the form
$\alpha_1 \beta \alpha_2$, recognizing $a^kxb^k$ for $k >j$.   Assume that $x$ occurs in $\alpha_1$.   Then $\alpha_1 \beta \alpha_2$
has a prefix $a^jx$, and cannot recognize $a^kxb^k$.   A similar argument holds if $\alpha_2$ contains $x$.
Hence $L_{axb}$ cannot be expressed by a PRS, even though $L_{ab} =\{a^ib^i : i \ge 0\}$
and $L_{ab} \cup L_{axb}$ can be expressed by a PRS.  
 
 We also suspect that non-deterministic CFLs are not expressible by a PRS. This is because the PRS definition does not allow the generation of non-deterministic finite automata (NDFAs). For example, consider the following PRS:  let $p^0$ and $p^1$ be the patterns that accept the characters
``0" and ``1" respectively, and $p^{00} = p^0 \compose p^0$ and $p^{11} = p^1 \compose p^1$.
Let $R_{pal}$ consist of the rules $\bot\ \rulearrow \ p^{00}$, $\bot \ \rulearrow \ p^{11}$, 
$p^{00} \ \rulearrow_s \ (p^0 \compose p^0) \connect p^{00}$, 
$p^{00} \ \rulearrow_s \ (p^0 \compose p^0) \connect p^{11}$, 
$p^{11} \ \rulearrow_s \ (p^1 \compose p^1) \compose p^{00}$, and 
$p^{11} \ \rulearrow_s \ (p^1 \compose p^1) \connect p^{11}$. 
At first, it seems that $L(R_{pal})$ is exactly the (non-deterministic) CFL of even-length palindromes over the alphabet $\{0,1\}$.  However, this is not actually the case: the rules $p^{00} \ \rulearrow_s \ (p^0 \compose p^0) \connect p^{00}$ and $p^{11} \ \rulearrow_s \ (p^1 \compose p^1) \connect p^{11}$ are never applicable, as grafting a copy of $p^{00}$ (or $p^{11}$) onto its own join state introduces a non-deterministic transition. 
We could extend the definition of a PRS to also create NDFAs; however, this can introduce ambiguity and complicates the inference algorithm.

\bibliographystyle{splncs04}
\bibliography{PRS}

\end{document}